\documentclass{article}[10pt]
\usepackage{fullpage}
\usepackage{amsmath}
\usepackage{amsthm}
\usepackage{amssymb}
\usepackage{mathtools}
\usepackage{verbatim}
\usepackage{color}
\usepackage{boxedminipage}
\usepackage{slashbox}
\usepackage{enumitem}
\usepackage[caption=false]{subfig}
\usepackage{url}
\usepackage{listings}
\usepackage{xcolor}
\usepackage{hyperref}

\newtheorem{theorem}{Theorem}
\newtheorem{definition}{Definition}
\newtheorem{lemma}{Lemma}

\setlist{noitemsep, leftmargin=*,topsep=2pt}

\mathchardef\mhyphen="2D
\newcommand{\Hyphen}{\mhyphen}

\newcommand{\OrderedFenceFunc}{\mathsf{ordered}_f}
\newcommand{\OrderedFunc}{\mathsf{ordered}}

\newcommand{\AddrDep}{<_{adep}}
\newcommand{\DataDep}{<_{ddep}}
\newcommand{\LSFOrd}{<_{ppomf}}
\newcommand{\DepOrd}{<_{ppod}}
\newcommand{\SameAddrOrd}{<_{pposa}}
\newcommand{\FenceOrd}{<_{ppof}}

\newcommand{\ProgOrd}{<_{po}}
\newcommand{\PreservePO}{<_{ppo}}
\newcommand{\MemOrd}{<_{mo}}
\newcommand{\ReadFrom}{\xrightarrow{}_{r\!f}}

\newcommand{\DoneTS}{\mathsf{doneTS}}
\newcommand{\AddrTS}{\mathsf{addrTS}}
\newcommand{\StDataTS}{\mathsf{sdataTS}}
\newcommand{\FromSt}{\mathsf{from}}
\newcommand{\MemAddr}{\mathsf{addr}}
\newcommand{\StData}{\mathsf{sdata}}
\newcommand{\LdVal}{\mathsf{ldval}}

\newcommand{\RobAddrDep}{<_{adep\Hyphen{}rob}}
\newcommand{\RobDataDep}{<_{ddep\Hyphen{}rob}}

\newcommand{\RobProgOrd}{<_{po\Hyphen{}rob}}
\newcommand{\RobNTPreservePO}{<_{ntppo\Hyphen{}rob}}

\newcommand{\Coh}{<_{co}}
\newcommand{\Fr}{\xrightarrow{}_{f\!r}}
\newcommand{\PoLoc}{<_{poloc}}
\newcommand{\RfInv}{\xrightarrow{}_{{r\!f}^{-1}}}
\newcommand{\Rfi}{\xrightarrow{}_{r\!f\!i}}
\newcommand{\Rfe}{\xrightarrow{}_{r\!f\!e}}
\newcommand{\Eco}{<_{eco}}

\newcommand{\IIEMemOrd}{<_{mo\Hyphen{}i2e}}
\newcommand{\IIEProgOrd}{<_{po\Hyphen{}i2e}}
\newcommand{\IIEPreservePO}{<_{ppo\Hyphen{}i2e}}

\DeclareMathOperator{\acyclic}{acyclic}

\begin{document}

\title{Weak Memory Models with Matching Axiomatic and Operational Definitions}

\author{\begin{tabular}{cccc}
        $\mathrm{Sizhuo\ Zhang}^1$ & $\mathrm{Muralidaran\ Vijayaraghavan}^1$ & $\mathrm{Dan\ Lustig}^2$ & $\mathrm{Arvind}^1$ \\
    \end{tabular}\\
    \begin{tabular}{cc}
        ${}^1$\{szzhang, vmurali, arvind\}@csail.mit.edu & ${}^2$dlustig@nvidia.com \\
        ${}^1$MIT CSAIL & ${}^2$NVIDIA \\
    \end{tabular}}
\date{}
\maketitle

\begin{abstract}
Memory consistency models are notorious for being difficult to define precisely, to reason about, and to verify.
More than a decade of effort has gone into nailing down the definitions of the ARM and IBM Power memory models, and yet there still remain aspects of those models which (perhaps surprisingly) remain unresolved to this day.
In response to these complexities, there has been somewhat of a recent trend in the (general-purpose) architecture community to limit new memory models to being (multicopy) atomic: where store values can be read by the issuing processor before being advertised to other processors.
TSO is the most notable example, used in the past by IBM 370 and SPARC-TSO, and currently used in x86. Recently (in March 2017) ARM has also switched to a multicopy atomic memory model, and the new RISC-V ISA and recent academic proposals such as WMM are pushing to do the same.

In this paper, we show that when memory models are atomic, it becomes much easier to produce axiomatic definitions, operational definitions, and proofs of equivalence than doing the same under non-atomic models.
The increased ease with which these definitions can be produced in turn allows architects to build processors much more confidently, and yet the relaxed nature of the models we propose still allows most or all of the performance of non-atomic models to be retained.
In fact, in this paper, we show that atomic memory models can be defined in a way that is parametrized by basic instruction and fence orderings.
Our operational vs.\@ axiomatic equivalence proofs, which are likewise parameterized, show that the operational model is sound with respect to the axioms and that the operational model is complete: that it can show any behavior permitted by axiomatic model.

For concreteness, we instantiate our parameterized framework in two forms.
First, we describe GAM (General Atomic Memory model), which permits intra-thread load-store reorderings.
Then, we show how forbidding load-store reordering (as proposed by WMM) allows the operational and axiomatic model to be even further simplified into one based on Instantaneous Instruction Execution (I2E).
Under I2E, each processor executes instructions in order and instantaneously, providing an even simpler model still for processors which do not implement load-store reordering.
We then prove that the operational and axiomatic definitions of I2E are equivalent as well.

\end{abstract}

\section{Introduction}

Interest in weak memory models stems from the belief that such models provide greater flexibility in implementation and thus, can lead to higher performance multicore microprocessors than those that support stronger memory models like \emph{Sequential Consistency} (SC) or \emph{Total Store Order} (TSO). 
However, extremely complicated and contentious definitions of POWER and ARM ISAs, which are the most important modern examples of industrially supported weak memory models, have generated somewhat of a backlash against weak memory models. 
As recently as 2017, a trend is emerging in which general-purpose processors are moving away from extremely weak (so-called ``non-atomic'') memory models and back towards simpler options which are much more tractable to understand and analyze.

Over the years, two competing memory model definition approaches have emerged.
One form is the \emph{operational model}, which is essentially an abstract machine that can run a program and directly produce its legal behaviors.
The other form is the \emph{axiomatic model}, which is a collection of constraints on legal program behaviors.
Each type has its own advantage.
Axiomatic models can use general-purpose combinatorial search tools like model checkers and SMT solvers to check whether a specific program behavior is allowed or disallowed, and they are useful for building computationally-efficient tools~\cite{alglave2014herding,memalloy,lustig2017automated}.
However they are not as suitable for inductive proofs that aim to build up executions incrementally, and many architects find them rather non-intuitive and a big departure from actual hardware. 
On the other hand, operational models are very natural representations of actual hardware behavior, and their small step semantics are naturally very well suited to building formal inductive proofs~\cite{Nienhuis:2016:OSC:2983990.2983997}.

Given the complementary natures of these two types of definitions, it would be ideal if a memory model could have an axiomatic definition and an operational definition which match each other.
Then different definitions can serve different use cases.
This is indeed the case for strong memory models like SC and TSO, but unfortunately, not so for weak memory models.
The research in weak memory models can then be classified into the following two categories:
\begin{enumerate}
    \item Build accurate axiomatic and operational models of existing architectures.
    \item Specify what memory models ought to look like: proposed memory models should be simple to understand with no obvious restrictions on implementations, and the equivalence of axiomatic and operational models may even be understood intuitively.
\end{enumerate}
While great efforts have been devoted to the first type of research to create models for commercial architectures like POWER and ARM, these models and proofs are still subject to subtle incompatibilities and frequent model revisions that invalidate the efforts~\cite{sarkar2011understanding,alglave2014herding,flur2016modelling,lahav2017repairing}.
More importantly, the veracity of these models is hard to judge because they are often based on information which is not public.
For example, the ARM operational model proposed by Flur et al.~\cite{flur2016modelling} allows many non-atomic memory behaviors that cannot be observed in any ARM hardware, and the paper claims that those behaviors are introduced to match the intentions of ARM's architects.
However, the recently released ARM ISA manual~\cite{armv8ar} clearly forbids those behaviors, invalidating the model completely.

This paper falls in the second category, and is motivated by the aim to reduce the complexity of commercial weak memory models.
The results in this paper are not purely academic -- the growing importance of the open source RISC-V ISA~\cite{riscv} has provided an opportunity to design a clean slate memory model.
The memory model for RISC-V is still being debated, and the members of the RISC-V community who are involved in the debate have expressed a strong desire for both axiomatic and operational definitions of the memory model.

In this paper, we present a framework which provides an axiomatic semantics, an operational semantics, and proofs of equivalence, and all in a way that is \emph{parameterized} by the basic instruction and fence orderings in a given model.
With our model, specifications and proofs are not nearly as fragile and subject to frequent breakage with every subtle tweak to a memory model definition.
Instead, the parameterization allows fence semantics to be simply and easily tweaked as needed.

\subsection{Contributions}
The main contribution of this paper is GAM, a general memory model for systems with atomic memory.
The model is parameterized by fences and basic instruction reorderings.
Both its operational and axiomatic definitions can be restricted to provide definitions of other simpler atomic memory models.
We provide proofs that the operational definition of GAM is \emph{sound} and \emph{complete} with respect to its \emph{axiomatic} definition.
We believe that GAM is the first memory model that allows load-store reordering and for which matching axiomatic and operational definitions have been provided.


On top of GAM, we show that GAM can be further simplified by simply preventing load-store reordering.
Such models can be described in terms of Instantaneous Instruction Execution (I2E), a model in which instructions execute instantaneously and in-order, with special memory buffers capturing the weak memory behaviors.
Furthermore, I2E models can additionally be parameterized by dependency orderings (under a commonly satisfied constraint), providing even more flexibility.
We provide proofs of equivalence for our axiomatic and operational definitions of I2E as well.

\noindent\textbf{Paper organization:} 
In Section \ref{sec:background}, we present three issues that complicate the definitions of weak memory models.
In Section \ref{sec:related}, we presented the related work.
In Section \ref{sec:GAM}, we present the axiomatic and operational definitions of our parameterized general atomic memory model GAM, along with the proofs of the equivalence of the two definitions.
In Section \ref{sec:COM}, we present an alternative axiomatic definition of GAM because this definition is better suited for using model checkers.
In Section \ref{sec:instance}, we show how GAM can be restricted to represent other simpler memory models.
In Section \ref{sec:I2E}, we show that if Load-Store reordering is disallowed, then the operational models can be described in the Instantaneous Instruction Execution manner and parameterized by dependency orderings.
Finally we end the paper with brief conclusions in Section \ref{sec:conclusion}.

\section{Memory Model Background}\label{sec:background}

In the following, we discuss the three specific challenges in defining a weak memory model such that it has matching operational and axiomatic definitions, and  explain briefly how we tackle the challenges.

\subsection{Atomic versus Non-atomic Memory}
Both ARM (until March 2017) and IBM Power use what is known as \emph{non-atomic memory} which does not have a universally-agreed-upon definition~\cite{alglave2014herding,maranget2012tutorial}. 
A major source of complication in weak model definitions stems from the use of non-atomic memory.
This lack of consensus makes it difficult to have matching definitions with non-atomic memory.
In this paper, we define memory models that use \emph{atomic memory}, or more precisely its variant which is known as \emph{multicopy atomic memory}. 
By atomic memory we mean a conceptual \emph{multiported monolithic memory} where loads and stores are executed instantaneously and a load $a$ returns the value of the latest store to address $a$.
Multicopy atomic memory lets a processor that generates a store bypass the value of that store to other newer loads in the same processor, before other processors may see that store value.   
Multicopy atomic memory captures the abstraction of a store buffer in the microarchitecture and is the underlying memory system for the popular TSO memory model used by Intel and AMD~\cite{sewell2010x86}.
In this paper we will use the term atomic memory and multicopy atomic memory interchangeably.

In the RISC-V debate a strong consensus has emerged that the memory model for RISC-V should depend only on atomic memory and therefore in this paper we will discuss only atomic memory models.

\subsection{Instruction Reorderings and Single-thread Semantics}\label{sec:inst-reorder}
Modern high-performance processors invariably execute instructions out of order (aka OOO processors) but they do it such that this reordering is transparent to a single threaded program.
However, in a multithreaded setting these instruction reorderings become visible. 
A major classification of memory models is along the lines of which (memory) instruction reorderings are permitted. 
For example, SC does not allow any reordering, while TSO allows a Load to be reordered with respect to previous Stores (i.e., it allows Store-Load reordering).
WMM~\cite{wmm}, Alpha, ARM, and Power also permit Store-Store and Load-Load reordering, provided the accesses are to different addresses, and all of these, except WMM, also permit Load-Store reordering\cite{armv8ar,power2013version,alpha1998,wmm}. 
The same address Store-Store reordering would clearly destroy the single thread semantics and thus, is prohibited.
The reason for disallowing the same address Load-Load reorderings is subtler, and a variation of WMM can be defined that indeed allows such a reordering.

However, it should be noted that all of SC, TSO and WMM have matching axiomatic and operational definitions, while to our knowledge Alpha and RMO have only axiomatic definitions.
This difference is likely to be caused by the added complexity of permitting Load-Store reordering, i.e., issuing a Store to the memory before all the previous Loads have completed.
A consequence of allowing Load-Store reordering is that the value a load gets in a multithreaded setting can depend upon a future store from the same thread.  
This complicates operational definitions. 
Load-Store reordering also complicates axiomatic semantics where a special axiom is often needed to disallow so-called \emph{out-of-thin-air} (OOTA) behavior~\cite{Boehm:2014:OGA:2618128.2618134,alpha1998}.
These two factors add to the difficulty of matching axiomatic and operational definitions. 

The \emph{General Atomic Memory} (GAM) model defined in this paper takes the challenge and allows all four reorderings (i.e., including Load-Store reordering).
To model Load-Store reordering, this paper provides an operational definition of GAM using unbounded Reorder Buffer (ROB) with speculative execution and atomic memory. 
(The memory system itself is not speculative, i.e., once a store has been issued it cannot be retracted.)
A similar mechanism has been used in the past to define the operational model for Power~\cite{sarkar2011understanding}, but there are separate concerns about that model which are discussed in Section~\ref{sec:related}.

\subsection{Fences for Writing Multithreaded Programs}
If we classify memory models based on instruction reordering only then for a given program, GAM allows more program behaviors than WMM, WMM allows more behaviors than TSO, and TSO allows more behaviors than SC.
More behaviors generally mean more flexibility in hardware implementation, however, a programmer needs a way to control instruction reorderings in order to write shared memory multithreaded programs.
The foundations of all multithreaded programming, from Dijkstra~\cite{ReadersAndWriters1965} and Lamport~\cite{lamport1979make} to current Java multithreaded libraries, is based on SC, that is, order-preserving interleaving of instructions in a multithreaded program.
Hence as a minimum, any ISA supporting a memory model weaker than SC must provide \emph{fence} instructions to make it possible to disallow instruction reorderings to enforce SC, if desired. 
Not surprisingly, different models require different types of fences and the execution cost of a fence varies from implementation to implementation. 

Fences are often explained in two entirely different ways.
One way is to define a fence simply to prevent reordering between loads and stores. 
For example, RMO and RISC-V have four individual fence components, FenceSS, FenceLS, FenceSL and FenceLL, to prevent reorderings between Store-Store, Load-Store, Store-Load and Store-Store, respectively (the actual names of fence instructions are different), and as many as fifteen fences can be formed by composing these options.
Such fences specify when two instructions in a dynamic instruction stream in a processor may not be reordered.
Of course, for complete specification, one also has to specify how fences may be reordered with respect to each other or how/whether, for example, FenceSS may be reordered with respect to a Load. 
This view of fences is only about reordering with in a processor and has nothing to do with the memory system.

Specifying how fences control instruction reordering is not sufficient to understand how programs behave.
We need to specify what the meaning of ``a store has completed'', i.e., when the value of a store becomes visible to loads in other processors or to a load in the same processor.
One needs to understand the details of the memory subsystem, such as presence of store buffers, write through caches, etc., to give precise meaning to fences.

The second type of fence definitions is usually explained in terms of their effect on memory.
For example, a Store-Release fence (alternatively known as a Commit) blocks the execution of the following stores until all the preceding instructions have completed.
Similarly, Load-Acquire fence (or Reconcile) blocks the execution of following instructions until all preceding loads are satisfied.
Similarly, there is Full-fence instruction that blocks the execution of all subsequent memory instructions until all the preceding memory instructions have completed.

In addition to subtle differences in the semantics of fences, there can be huge differences in performance penalty of using different types of fences.
For example, a full fence may be overkill in an algorithm where it may be sufficient to keep two sequential stores from being reordered.
Insertion of fences in a multithreaded program by the programmer or the compiler writer is one of the thorniest problems related to weak memory models. 
If too many unnecessary fences are inserted in a program then it would show poor performance, and in the extreme case the whole purpose of having a weak memory model would be lost. 
If too few fences are inserted, then the meaning of a program may change by admitting new behaviors which may not be acceptable. 
The debugging of multithreaded programs is a difficult task in the best of times, insertion of fences creates the possibility of including even more silent bugs which may manifest under very peculiar scheduling conditions. 
Automatic insertion of fences by a compiler for the programming model such as the one embodied in C11 may be feasible but that memory model of C11 is already based on some cost assumptions of various fences, creating a catch-22 situation~\cite{c11}.

The lack of agreement on the set of fences and the nuances between different fences all add to the difficulty of matching axiomatic and operational definitions.
To address these problems, GAM restricts itself to a very simple atomic memory model where there is no ambiguity about when a value is visible to other processors.
Such atomic memory automatically avoids many of the thorniest difficulties (such as cumulativity~\cite{alglave2014herding}) in the definitions of fences.
Since there is still no clear consensus on which set of fences gives the best tradeoff between ease of use and performance, we have parameterized the GAM model with the type of fences.
The axiomatic definition, operational definition and the proofs of equivalence are all also parameterized by the type of fences.


\section{Related Work}\label{sec:related}

\newcommand{\etal}{et al.}
\newcommand{\etc}{etc.}

Lamport's paper on SC \cite{lamport1979make} is probably the first formal definition, both axiomatic and operational, of a memory model. In the nineties, three different weak memory models were defined axiomatically for SUN's Sparc processors: TSO, PSO and RMO~\cite{sparc1992sparcv8,weaver1994sparc}. A weak memory model for DEC Alpha was also specified axiomatically in the same time frame~\cite{alpha1998}.
Until a decade ago, however, there was no effort to specify weak memory models operationally or match axiomatic specifications to operational models. In this context, papers by Sarkar \etal \cite{Sarkar:2009:SXM:1594834.1480929}, Sewell \etal
\cite{sewell2010x86} and Owens \etal{} \cite{owens2009better} are very important because they showed that the axiomatic specification of TSO is exactly equivalent to an operational model using store buffers connected to I2E processors and atomic memory.

Until recently, weak memory models have not been defined prior to ISA
implementation and have been documented by manufacturers only in an ad
hoc manner using a combination of natural language and litmus tests.
Not surprisingly, such ``definitions'' have had to be revised as
implementations have changed, revealing new corner cases of behaviors.
Over the last decade, several studies have been performed, mostly by
academic researchers, to determine the allowed and disallowed behavior
of several commercial microprocessors, with the goal of creating
formal models to explain the observed behaviors.  These studies have
been done on real machines by running billions of instructions and
recording the observations (just like studying any natural
phenomenon).  Then, with extra inputs from hardware designers, a model
is constructed that tries to satisfy all these observations.  
For example, 
Sarkar \etal{} specified an operational model for POWER
\cite{sarkar2011understanding, sarkar2012synchronising}, using a
non-atomic memory.
Later, Mador-Haim \etal{} \cite{mador2012axiomatic}
developed an axiomatic model for POWER and proved that it matches the
earlier operational model. 
Alglave \etal{} \cite{alglave2009semantics, Alglave2011,
  alglave2012formal, alglave2014herding, alglave2013software} give
axiomatic specifications for ARMv7 and POWER using the Herd framework;
Flur \etal{} \cite{flur2016modelling} give operational specification
for ARMv8.

However, there has been some dispute if the
operational model of POWER models actual POWER processors accurately
\cite{alglave2014herding}. We attribute the reason for the potential
errors to be the inherent complexity of the operational model because of the use of non-atomic memory.
Alglave models are not sufficiently grounded in operational models and face the problem of being too liberal. The model may admit behaviors which cannot be observed in any implementation. Such models can also lead to insertion of unnecessary fences in a program. 
We think it is important to have matching operational and axiomatic models.

Researchers have also proposed several other consistency models:
Processor Consistency~\cite{goodman1991cache}, Weak
Consistency~\cite{dubois1986memory}, RC~\cite{gharachorloo1990memory},
CRF~\cite{shen1999commit}, Instruction Reordering + Store
Atomicity~\cite{arvind2006memory}.  
The tutorials by Adve \etal{} \cite{adve1996shared} and by Maranget
\etal{} \cite{maranget2012tutorial} provide relationships among some
of these models.

Researchers have also proposed architectural mechanisms for implementing SC \cite{lamport1979make} efficiently
\cite{gharachorloo1991two,ranganathan1997using,guiady1999sc+,gniady2002speculative,ceze2007bulksc,wenisch2007mechanisms,blundell2009invisifence,singh2012end,lin2012efficient,gope2014atomic}. 
Several of these architectural mechanisms are interesting in their own right and applicable to reducing power consumption, however,
so far commercial processor vendors have shown little interest in adopting stricter memory models.

Recently, there is a splurge of activity in trying to specify
semantics of concurrent languages: C/C++
\cite{c++n4527,boehm2008foundations,batty2011mathematizing,Batty:2016:OSA:2914770.2837637,Kang:2015:FCM:2737924.2738005},
Java \cite{manson2005java,cenciarelli2007java, maessen2000improving}.
These models are specified axiomatically, and allow load-store reordering.
For C++, there has been work to specify an equivalent operational model
\cite{operationalC++}.





\section{General Atomic Memory Model (GAM)}\label{sec:GAM}

In this section, we introduce GAM, an atomic memory model framework parametrized by how the memory model enforces following two types of orderings:
\begin{enumerate}
    \item Memory instruction ordering: the ordering between two memory instructions, i.e., the commonly referred load-load, load-store, store-store and store-load orderings.
    \item Fence ordering, the ordering between a fence and a memory instruction or between two fences.
\end{enumerate}
We refer to the combination of the above two orderings as \emph{memory/fence ordering}.
GAM uses a function $\OrderedFunc(I_{old}, I_{new})$ to represent memory/fence ordering, and this function is used in both the axiomatic and operational definitions of GAM.
$\OrderedFunc(I_{old}, I_{new})$ returns true when the older instruction $I_{old}$ should be ordered before the younger instruction $I_{new}$ according to the memory instruction ordering or fence ordering enforced by the memory model.
For example, Table~\ref{tab:order-tso} shows the $\OrderedFunc(I_{old}, I_{new})$ table for TSO, which has only one type of fence.
The only memory ordering that is not enforced by TSO is the store-load ordering, as represented by the false entry (St, Ld).
As a more complex example, Table~\ref{tab:order-rmo} shows the $\OrderedFunc(I_{old}, I_{new})$ table for RMO.
In RMO, all four memory instruction orderings are relaxed, as shown by the false entries (Ld, Ld), (Ld, St), (St, Ld) and (St, St).
The four fences are used to enforce each type of orderings respectively.
For example, the true entries (FenceLS, St) and (Ld, FenceLS) means that FenceLS is ordered before younger stores and is ordered after older loads, thus enforce load-to-store ordering.
The fences are even unordered with respect to each other.
As a framework, given an $\OrderedFunc$ function (such as Table~\ref{tab:order-tso} or \ref{tab:order-rmo}), GAM can produce equivalent axiomatic and operational models that enforce the memory/fence orderings represented by the $\OrderedFunc$ function.

\begin{table}[!htb]
    \centering
    \begin{tabular}{|l|l|l|l|}
        \hline
        \backslashbox{$I_{old}$}{$I_{new}$} & Ld    & St    & Fence \\ \hline
        Ld                                  & True  & True  & True  \\ \hline
        St                                  & False & True  & True  \\ \hline
        Fence                               & True  & True  & True  \\ \hline
    \end{tabular}
    \caption{Orderings for TSO memory instructions and fences: $\OrderedFunc_{TSO}(I_{old}, I_{new})$}\label{tab:order-tso}
\end{table}
\begin{table}[!htb]
    \centering
    \begin{tabular}{|l|l|l|l|l|l|l|}
        \hline
        \backslashbox{$I_{old}$}{$I_{new}$} & Ld    & St    & FenceLL   & FenceLS   & FenceSL   & FenceSS \\ \hline
        Ld                                  & False & False & True      & True      & False     & False \\ \hline
        St                                  & False & False & False     & False     & True      & True  \\ \hline
        FenceLL                             & True  & False & False     & False     & False     & False \\ \hline
        FenceLS                             & False & True  & False     & False     & False     & False \\ \hline
        FenceSL                             & True  & False & False     & False     & False     & False \\ \hline
        FenceSS                             & False & True  & False     & False     & False     & False \\ \hline
    \end{tabular}
    \caption{Orderings for RMO memory instructions and fences: $\OrderedFunc_{RMO}(I_{old}, I_{new})$}\label{tab:order-rmo}
\end{table}

It should be noted that the memory/fence ordering cannot fully describe a memory model.
The following three aspects are not captured by the $\OrderedFunc$ function:
\begin{enumerate}
    \item Load value: a memory model must specify which store values a load may read.
    \item Dependency ordering: most memory models order two instructions if the younger instruction is dependent on the older instruction in certain ways.
    \item Same-address ordering: even when the memory/fence ordering does not apply to two memory instructions for the same address, a memory model may still order them for the correctness of single-threaded programs.
\end{enumerate}
The GAM definition given in this section is not parametrized in terms of the above three aspects.
Later in Section~\ref{sec:instance}, we will show how to tweak the definition of GAM to derive memory models with a different dependency ordering or a different same-address ordering.
The way to determine load values should be common across all multicopy atomic memory models, so we do not bother changing that.
In the following, we give axiomatic and operation definitions of GAM and the equivalence proof.
When we use examples to explain our definitions, we assume the $\OrderedFunc$ function in Table~\ref{tab:order-rmo}, i.e., with all four memory instruction reorderings and relaxed fences.

We present the axiomatic definition before the operational definition but these definitions can be read in any order.

\subsection{Axiomatic Definition of GAM}
\label{sec:gam:axiom}

The axiomatic definition of GAM takes three relations as input: program order ($\ProgOrd$), read-from relations ($\ReadFrom$) and memory order ($\MemOrd$).
The program order ($\ProgOrd$) is a per-processor total order and represents the order in which the instructions are committed in that processor.
A read-from edge specifies that a load reads from a particular store; $\ReadFrom$ points from a store instruction to a load instruction for the same address, with the load getting the same value that is written by the store.
The memory order ($\MemOrd$) is a total order of all memory instructions in all processors.
Intuitively, $\MemOrd$ specifies the order of when each memory accesses are performed globally.

The axiomatic model checks $\ProgOrd$, $\ReadFrom$ and $\MemOrd$ against a set of axioms.
If all the axioms are satisfied, then the program behavior given by $\ProgOrd$ is allowed by the memory model.

It should be noted that $\ProgOrd$ is the observable program behavior, while $\ReadFrom$ and $\ProgOrd$ are just a \emph{witness} which cannot be observed directly.
To justify that a program behavior is allowed by GAM, we only need to find one witness (i.e., $\langle \ReadFrom, \MemOrd\rangle$) that satisfies all the axioms.
To prove that a program behavior is disallowed by GAM, we must show that there is no witness that can satisfy all the axioms simultaneously.

In order to describe the axioms, we first define \emph{preserved program order} ($\PreservePO$), which is computed from $\ProgOrd$.
$\PreservePO$ captures the constraints on the out-of-order (OOO) execution of instructions in each processor (locally).
Thus, a property of $\PreservePO$ is that if $I_1 \PreservePO I_2$ then $I_1 \ProgOrd I_2$ where $I_1$ and $I_2$ are instructions.

As will become clear $\ProgOrd$ by itself cannot reflect the constraints on the memory system and the interaction between processors.  
These constraints are expressed by the separate memory axioms of GAM.
In the following, we first define how to compute $\PreservePO$ from $\ProgOrd$, and then give the memory axioms of GAM.

\subsubsection{Definition of Preserved Program Order $<_{ppo}$ for GAM}\label{sec:ppo}

We define $\PreservePO$ in three parts.
The first part is the \emph{preserved memory/fence order} ($\LSFOrd$) which is captured by the $\OrderedFunc$ function.
The second part is the \emph{preserved dependency order} ($\DepOrd$), which includes branch dependencies, address dependencies, data dependencies, etc.
The last part is the \emph{preserved same-address order} ($\SameAddrOrd$), i.e., the ordering of memory instructions for the same address. 
Finally $\PreservePO$ is defined as the transitive closure of $\LSFOrd$, $\DepOrd$ and $\SameAddrOrd$.

\noindent\textbf{Definition of preserved memory/fence order $\LSFOrd$ for GAM:}
The preserved memory/fence order is fully described by the $\OrderedFunc$ function.
\begin{definition}[Preserved memory/fence order $\LSFOrd$\label{def:ppo-mem-fence}]
    $I_1 \FenceOrd I_2$ iff $I_1$ and $I_2$ both are memory or fence instructions, and $I_1 \ProgOrd I_2$, and $\OrderedFunc(I_1,I_2)$ is true.
\end{definition}

\noindent\textbf{Definition of preserved dependency order $<_{ppod}$ for GAM:}
We first give some basic definitions that are used to define dependency orderings precisely (all definitions ignore the PC register and the zero register): 
\begin{definition}[RS: Read Set]
    $RS(I)$ is the set of registers an instruction $I$ reads.
\end{definition}

\begin{definition}[WS: Write Set]
    $WS(I)$ is the set of registers an instruction $I$ can write.
\end{definition}

\begin{definition}[ARS: Address Read Set]
    $ARS(I)$ is the set of registers a memory instruction $I$ reads to
    compute the address of the memory operation.
\end{definition}

\begin{definition}[data-dependency $\DataDep$ \label{def:data-dep}]
    $I_1 \DataDep I_2$ if $I_1 \ProgOrd I_2$ and $WS(I_1) \cap
    RS(I_2) \neq \emptyset$ and there exists a register
    $r$ in $WS(I_1) \cap RS(I_2)$ such that there is no instruction
    $I$ such that $I_1 \ProgOrd I \ProgOrd I_2$ and $r \in WS(I)$.
\end{definition}

\begin{definition}[addr-dependency $\AddrDep$ \label{def:addr-dep}]
    $I_1 \AddrDep I_2$ if $I_1 \ProgOrd I_2$ and $WS(I_1) \cap
    ARS(I_2) \neq \emptyset$ and there exists a register
    $r$ in $WS(I_1) \cap ARS(I_2)$ such that there is no instruction
    $I$ such that $I_1 \ProgOrd I \ProgOrd I_2$ and $r \in WS(I)$.
\end{definition}

Note that data-dependency includes addr-dependency,
i.e., $I_1 \AddrDep I_2$ $\implies$ $I_1 \DataDep I_2$.

Now we define $\DepOrd$, which essentially says that the data-dependencies must be observed, stores should not execute until the preceding branches have been resolved, the execution of stores should be constrained by instructions on which prior memory instructions are
address dependent, and in the case of a load following a store to the same address, the execution of the load should be constrained by instructions which produce the store's data.

\begin{definition}[Preserved dependency order $\DepOrd$
\label{def:ppo-dep}]
    $I_1 \DepOrd I_2$ if either
    \begin{enumerate}
        \item \label{ppo:ddep} $I_1 \DataDep I_2$, or
        \item \label{ppo:br->st} $I_1 \ProgOrd I_2 $, and $I_1$ is a branch, and $I_2$ is a store, or
        \item \label{ppo:adep+po} $I_2$ is a store instruction, and there
        exists a memory instruction $I$ such that $I_1 \AddrDep I \ProgOrd I_2$,
        or
        \item \label{ppo:ddep+sa} $I_2$ is a load instruction, and there
        exists a store $S$ to the same address such that $I_1 \DataDep S
        \ProgOrd I_2$, and there is no other \emph{store} for the same
        address between $S$ and $I_2$.
    \end{enumerate}
\end{definition}

In the above definition, cases~\ref{ppo:ddep} and \ref{ppo:br->st} are straightforward; we discuss the rest of the cases below.

Case~\ref{ppo:adep+po} is about a subtle dependency caused by an address dependency and is illustrated by the example in
Figure~\ref{fig:adep-po-example}.  
If $I_3$ (store) is allowed to be issued before $I_1$, then the earlier load ($I_2$) may end up reading its own future store in case $I_1$ returns value $r_1=b$.

\begin{figure}[!htb]
    \centering
    \begin{tabular}{|l|}
        \hline
        $I_1:$ $r_1$ = Ld $a$\\
        $I_2:$ $r_2$ = Ld $r_1$\\
        $I_3:$ St $b$ = 1\\
        \hline
    \end{tabular}
    \caption{Example for case~\ref{ppo:adep+po}}\label{fig:adep-po-example}
\end{figure}

Case~\ref{ppo:ddep+sa} is about another subtle dependency when data is
transfered not by registers but by local bypassing.  In
Figure~\ref{fig:ddep-sa-example}, $I_3$ must be issued after $I_1$.
Otherwise, in case $I_3$ is issued before $I_1$, $I_3$ must bypass
from $I_2$.  However, the data of $I_2$ is still unknown at that time.
\begin{figure}[!htb]
    \centering
    \begin{tabular}{|l|}
        \hline
        $I_1:$ $r_1$ = Ld $a$ \\
        $I_2:$ St $b$ = $r_1$ \\
        $I_3:$ $r_2$ = Ld $b$ \\ 
        \hline
    \end{tabular}
    \caption{Example for case~\ref{ppo:ddep+sa}}\label{fig:ddep-sa-example}
\end{figure}

\noindent\textbf{Definition of preserved same-address order $\SameAddrOrd$ for GAM:}
Next we give the definition of $\SameAddrOrd$, which captures the orderings between memory instructions for the same address.
\begin{definition}[Preserved same-address order $\SameAddrOrd$\label{def:ppo-same-addr}]
    $I_1 \SameAddrOrd I_2$ if either
    \begin{enumerate}
        \item \label{ppo:ld->st} $I_1 \ProgOrd I_2$, and $I_1$ is a load and $I_2$ is
        a store to the same address, or
        \item \label{ppo:st->st} $I_1 \ProgOrd I_2$, and both $I_1$ and $I_2$ are store
        instructions for the same address, or
        \item \label{ppo:ld->ld} $I_1 <_{po} I_2$ and both $I_1$ and $I_2$ are load
        instructions for the same address with no intervening \emph{store} to the
        same address.
    \end{enumerate}
\end{definition}

The above definition explicitly excludes the enforcement of ordering of a store followed by a load to the same address.
Otherwise our model would be stricter than TSO in some cases.
Case~\ref{ppo:ld->ld} requires that loads for
the same address without store to the same address in between to be
issued in order. For example, all instructions in
Figure~\ref{fig:lda-lda-example} must be issued in order.

\begin{figure}[!htb]
  \centering
  \begin{tabular}{|l|}
    \hline
    $I_1:$ $r_1$ = Ld $a$ \\
    $I_2:$ $r_2$ = Ld ($b+r_1-r_1$) \\
    $I_3:$ $r_3$ = Ld $b$ \\
    $I_4:$ $r_4$ = Ld ($c+r_3-r_3$) \\
    \hline
  \end{tabular}
  \caption{Example for case~\ref{ppo:ld->ld}}\label{fig:lda-lda-example}
\end{figure}

It should be noted that the choice to enforce this same-address load-load ordering in GAM is kind of arbitrary, because we do not see any decisive argument to support either enforcing or relaxing this ordering.
On the one hand, implementations that execute loads for the the same address out of order will not violate single-thread correctness, and do not need the extra hardware to enforce this load-load ordering.
On the other hand, programmers may expect memory models to have the per-location SC property~\cite{cantin2003complexity}, i.e., all memory accesses for a single address appear to be sequentially consistent, and enforcing this same-address load-load ordering is an easy way to provide the per-location SC property.
The Alpha memory model~\cite{alpha1998} is the same as GAM in enforcing this ordering, while the RMO memory model~\cite{weaver1994sparc} chooses to relax this ordering completely.
In Section~\ref{sec:compare-rmo} programmers would like memory models to have the per-location SC property~\cite{cantin2003complexity}.
It should be noted that ARMv8.2 makes yet another choice in same-address load-load ordering which we will explain in Section~\ref{sec:compare-arm}.

Finally, we define $\PreservePO$ as the transitive closure of $\DepOrd$, $\SameAddrOrd$ and $\LSFOrd$.
\begin{definition}[Preserved program order $\PreservePO$\label{def:ppo}]
    $I_1 \PreservePO I_2$ if either
    \begin{enumerate}
        \item $I_1 \LSFOrd I_2$, or
        \item $I_1 \DepOrd I_2$, or
        \item $I_1 \SameAddrOrd I_2$, or
        \item \label{ppo:trans} there exists an instruction $I$ such that $I_1
        \PreservePO I$ and $I <_{ppo} I_2$.
    \end{enumerate}
\end{definition}

\subsubsection{Memory Axioms of GAM}

GAM has the following two axioms (the notation $\max_{mo}$ means to find the youngest instruction in $\MemOrd$):
\begin{itemize}
    \item \textbf{Axiom Inst-Order:} If $I_1 \PreservePO I_2$, then $I_1 \MemOrd I_2$.
    \item \textbf{Axiom Load-Value:}
    \begin{displaymath}
    \mathsf{St}\ a\ v \ReadFrom \mathsf{Ld}\ a \Rightarrow
    \mathsf{St}\ a\ v = \max_{mo}\{
    \mathsf{St}\ a\ v'\ |\ \mathsf{St}\ a\ v' \ProgOrd
    \mathsf{Ld}\ a\ \vee\ \mathsf{St}\ a\ v' \MemOrd \mathsf{Ld}\ a \}
    \end{displaymath}
\end{itemize}
The first axiom says that $\MemOrd$ must respect $\PreservePO$.
An interpretation of this axiom is that the local ordering constraints on executing two memory instructions in the processor must be preserved when these two memory accesses are performed globally.
The second axiom specifies the store that a load should read given $\MemOrd$ and $\ProgOrd$.
Intuitively, each store overshadows previous stores to the same address and thus, a load should not be able to read overshadowed values. 
The only complication is because of bypassing: a load may read one of its own store values before it is advertised, which means  a later load in other processors may still read the globally advertised store value in the memory.   
More precisely, the set of stores that are visible to a load consists of stores that either precede the load in $\ProgOrd$ or perform globally before the load does (i.e., precede the load in $\MemOrd$).
The store read by the load must be visible to the load, and cannot be overshadowed (in $\MemOrd$) by another store which is also visible to the load.

\subsection{An Operational Definition of GAM}
\label{sec:operational}

The operational model of GAM consists of $n$ processors $P_1\ldots P_n$ and a monolithic memory $m$.
Each processor $P_i$ consists of an ROB and a PC register.
The PC register contains the address of the next instruction to be fetched into ROB.
When an instruction is fetched, if the instruction is a branch, we predict the branch target address and update the PC register speculatively; otherwise we simply increment the PC register.
Each instruction in the ROB has a \emph{done} bit.
(We refer to an instruction as done if the
done bit is true, and as not done otherwise.)  Though
instructions that have been marked as done can be removed from the
ROB, we will not bother with this detail.

At each step of the execution one of the instructions marked as
not-done in the ROB of a processor $P_i$ is selected and executed and
(sometimes) marked as done. There is often a \emph{guard} condition
associated with the execution of an instruction, and an instruction
can be executed only if the guard is true. As will become clear soon
that sometimes the execution of an instruction cannot proceed even when its guard is true.

Our axiomatic model, permits very aggressive execution of load
instructions but it also requires that consecutive loads to the same address be done in
order.  If the operational model executed load instructions only when
the address for its preceding memory instructions were known, then we will not
be able to capture all the behaviors allowed by the axiomatic
model. Thus, in the operational model, we let a load execute even
before all the addresses of preceding memory instructions are known,
and then later kill a done load if an older memory instruction
happens to get the same address.  The kill of a load instruction
means that all the instruction younger than the killed load,
including that load itself, are discarded from the ROB, and the PC register is updated to make
instruction fetch begin by refetching the killed load instruction.

In order to implement these speculative loads, we need an additional
\emph{address-available} state bit in the ROB for each memory
instruction. This bit indicates when the address calculation has been
completed. Initially this bit is not set.

We need to know if the source operands of an instruction are available
in order to execute the instruction.  If the operand is
specified as a source register $r$, then its availability is
determined by searching the ROB from the current instruction slot
towards older instructions until the first slot containing $r$ as the
destination register. (The search always terminates because we assume
that the ROB has been initialized with instructions that set initial register values).
If the slot containing the destination register is marked as
done then the operand is assumed to be available, otherwise not.

This operational model is also parametrized by the memory/fence ordering.
That is, it uses $\OrderedFunc$ function to control when a memory or fence instruction can be marked as done.

In the following we specify how to execute an instruction in the ROB
according to the preserved program order definition given in Section
\ref{sec:ppo}.  Each operational rule has a guard and a specified
action.
\begin{itemize}
    \item \textbf{Rule Fetch:} \\
    \emph{Guard:} True. \\
    \emph{Action:} Fetch a new instruction from the address stored in the PC register.
    Add the new instruction into the tail of ROB.
    If the new instruction is a branch, we predict the branch target address of the branch, update PC to be the predicted address, and record the predicted address in the ROB entry of the branch; otherwise we increment PC.
    
    \item \textbf{Rule Execute-Reg-to-Reg:} Execute a reg-to-reg instruction $I$. \\
    \emph{Guard:} $I$ is marked not-done and all source operands of $I$ are ready. \\
    \emph{Action:} Do the computation, record the result in the ROB slot, and mark $I$ as done.
    
    \item \textbf{Rule Execute-Branch:} Execute a branch instruction $I$. \\
    \emph{Guard:} $I$ is marked not-done and all source operands of $I$ are ready. \\
    \emph{Action:} Compute the branch target address and mark $I$ as done.
    If the computed target address is different from the previously predicted address (which is recorded in the ROB entry), then we kill all instructions which are younger than $I$ in the ROB (excluding $I$).
    That is, we remove those instructions from the ROB, and update the PC register to the computed branch target address.
    
    \item \textbf{Rule Execute-Fence:} Execute a fence instruction $I$. \\
    \emph{Guard:} $I$ is marked not-done, and for each older (memory or fence) instruction $I'$ such that $\OrderedFunc(I', I)$ is true, $I'$ is done. \\
    \emph{Action:} Mark $I$ as done.
    
    \item \textbf{Rule Execute-Load:} Execute a load instruction $I$ for address $a$. \\
    \emph{Guard:} $I$ is marked not-done, and the address-available bit is set to available, and for each older (memory or fence) instruction $I'$ such that $\OrderedFunc(I', I)$ is true, $I'$ is done.\\
    \emph{Action:} Search the ROB from $I$ towards the oldest instruction for the first not-done memory instruction with address $a$:
    \begin{enumerate}
        \item If a not-done load to $a$ is found then instruction $I$ cannot be executed, i.e., we do nothing.
        \item If a not-done store to $a$ is found then if the data for the store is ready, then execute $I$ by bypassing the data from the store, and mark $I$ as done; otherwise, $I$ cannot be executed.
        \item If nothing is found then execute $I$ by reading $m[a]$, and mark $I$ as done.
    \end{enumerate}
    
    \item \textbf{Rule Compute-Store-Data:} compute the data of a store instruction $I$. \\
    \emph{Guard:} the source registers for the data computation are ready. \\ 
    \emph{Action:} Compute the data of $I$ and record it in the ROB slot.
    
    \item \textbf{Rule Execute-Store:} Execute a store $I$ for address $a$. \\
    \emph{Guard:} $I$ is marked not-done and in addition all the following conditions must be true:
    \begin{enumerate}
        \item The address-available flag for $I$ is set,
        \item The data of $I$ is ready,
        \item For each older (memory or fence) instruction $I'$ such that $\OrderedFunc(I', I)$ is true, $I'$ is done,
        \item All older branch instructions are done,
        \item \label{guard:addr->st} All older loads and stores have their address-available flags set,
        \item All older loads and stores for address $a$ are done.
    \end{enumerate}
    \emph{Action:} Update $m[a]$ and mark $I$ as done.
    
    \item \textbf{Rule Compute-Mem-Addr:} Compute the address of a load or store instruction $I$. \\
    \emph{Guard:} The address-available bit is not set and the address operand is ready with value $a$\\
    \emph{Action:} We first set the address-available bit and record the address $a$ into the ROB entry of $I$.
    Then we search the ROB from $I$ towards the youngest instruction (excluding $I$) for the first memory instruction with address $a$. 
    If the instruction found is a done load, then we kill that load and all instructions that are younger than the load in the ROB.
    That is, we remove the load and all younger instructions from the ROB, and set the PC register to the instruction-fetch address of the load.
    Otherwise no instruction needs to be killed.
\end{itemize}

\subsection{Soundness: GAM Operational model $\subseteq$ GAM Axiomatic Model}\label{sec:op<axiom}

The goal is to show that for any execution of the operational model, we can construct $\langle \ProgOrd, \MemOrd, \ReadFrom\rangle$ which satisfies the GAM axioms and has the same program behavior as the operational execution.
To do this, we need to introduce some ghost states to the operational model, and show invariants that hold after every step in the operational model.

In the operational model, we assume there is a (ghost) global time which is incremented whenever a rule fires.
We also assume each instruction $I$ in an ROB has the following ghost states which are accessed only in the proofs (all states start as $\top$):
\begin{itemize}
    \item $I.\DoneTS$: Records the current global time when a rule $R$ fires and marks $I$ as done.
    \item $I.\AddrTS$: Records the current global time for memory instruction $I$ when a Compute-Mem-Addr rule $R$ fires to compute the address of $I$.
    \item $I.\StDataTS$: Records the current global time for a store instruction $I$, when a Compute-Store-Data rule $R$ fires to compute the store data of $I$.
    \item $I.\FromSt$: Records the store read by $I$ if $I$ is a load.
    That is, the store is either the not-done store $I$ bypasses from or the done store with the maximum $\DoneTS$ among all done stores for $a$ when $I$ is marked as done.

\end{itemize}
In the final proof, we will use the states at the end of the operational execution to construct the axiomatic edges.
$\ProgOrd$ will be constructed by the order of instructions in ROB, $\ReadFrom$ will be constructed by the $\FromSt$ states of loads, and $\MemOrd$ will be constructed by the order of $\DoneTS$ timestamps of all memory instructions.

For convenience, we use $I.\LdVal$ to denote the load value if $I$ is a load, use $I.\MemAddr$ to denote the memory access address if $I$ is a memory instruction, and use $I.\StData$ to denote the store data if $I$ is a store.
These fields are $\top$ if the corresponding values are not available.

Given the model state at any time in the execution of the operational model, we can define the program order $\RobProgOrd$, data-dependency order $\RobDataDep$, address-dependency order $\RobAddrDep$, and a new relation $\RobNTPreservePO$ which is similar to the preserved program order. (we add suffix $rob$ to distinguish from the definitions in the axiomatic model):
\begin{itemize}
    \item $\RobProgOrd$: Instructions $I_1 \RobProgOrd I_2$ iff both $I_1$ and $I_2$ are in the same ROB and $I_1$ is older than $I_2$ in the ROB.
    \item $\RobDataDep$: $I_1 \RobDataDep I_2$ iff $I_1\RobProgOrd I_2$ and $I_2$ needs the result of $I_1$ as a source operand.
    \item $\RobAddrDep$: $I_1 \RobAddrDep I_2$ iff $I_1\RobProgOrd I_2$, and $I_2$ is a memory instruction, and $I_2$ needs the result of $I_1$ as a source operand to compute the memory address to access.
    \item $\RobNTPreservePO$: $I_1 \RobNTPreservePO I_2$ iff $I_1 \RobProgOrd I_2$ and at least one of the following conditions hold:
    \begin{enumerate}
        \item $I_1 \RobDataDep I_2$.
        \item $I_1$ is a branch, and $I_2$ is a store.
        \item $I_2$ is a store, and there exists a memory instruction $I$ such that $I_1\RobAddrDep I\RobProgOrd I_2$.
        \item $I_2$ is a load with $I_2.\MemAddr = a \neq \top$, and there exists a store $S$ with $S.\MemAddr = a$, and $I_1 \RobDataDep S \RobProgOrd I_2$, and there is no store $S'$ such that $S'.\MemAddr = a$ and $S \RobProgOrd S' \RobProgOrd I_2$.
        \item $I_1$ is a load with $I_1.\MemAddr = a \neq \top$, and $I_2$ is a store with $I_2.\MemAddr = a$.
        \item Both $I_1$ and $I_2$ are stores with $I_1.\MemAddr = I_2.\MemAddr = a \neq \top$.
        \item Both $I_1$ and $I_2$ are loads with $I_1.\MemAddr = I_2.\MemAddr = a \neq \top$, and there is no store $S$ such that $S.\MemAddr = a$ and $I_1 \RobProgOrd S \RobProgOrd I_2$.
        \item $\OrderedFunc(I_1, I_2)$ is true.
    \end{enumerate}
\end{itemize}
It should be noted that the way to compute $\RobNTPreservePO$ from $\RobProgOrd$ is almost the same as the way to compute $\PreservePO$ from $\ProgOrd$ except for two differences.
The first difference is that $\RobNTPreservePO$ is not made transitively closed; this is for simplifying the proof to some degree.
The second difference is that in case the definition needs the address of memory instructions, $\RobNTPreservePO$ ignores memory instructions which have not computed their addresses.
Since the address of every memory instruction will be computed at the end of the operational execution, the second difference will diminish by that time.
Since $\ProgOrd$ is defined by the $\RobProgOrd$ at the end of the operational execution,  $\PreservePO$ will be the transitive closure of $\RobNTPreservePO$ at the end of the operational execution.

With the above definitions, we give the invariants of any operation execution in Lemma~\ref{lem:rob:inv}.
Invariant~\ref{inv:rob:ppo} is a similar statement to the Inst-Order axiom, and will become exactly the same as that axiom at the end of the operational execution.
Invariants~\ref{inv:rob:addr} and \ref{inv:rob:st-data} captures the ordering effects of dependencies carried to the computation of memory address and store data.
Invariant~\ref{inv:rob:addr-st} captures guard~\ref{guard:addr->st} of the Execute-Store rule, and is also related to case~\ref{ppo:adep+po} of Definition~\ref{def:ppo-dep} for $\PreservePO$.
Invariant~\ref{inv:rob:kill-done-st} is an important property saying that stores are never written to the shared memory speculatively, so the model does not need any system-wide rollback.
Invariant~\ref{inv:rob:mem-val} constrains the current monolithic memory value.
Invariant~\ref{inv:rob:rf} constrains the store read by a load, and in particular, invariant~\ref{prop:rob:rf:done-max} will become the Load-Value axiom at the end of the operation execution.
The detailed proof can be found in Appendix~\ref{sec:gam_axi_contain_op}.
\begin{lemma}\label{lem:rob:inv}
    The following invariants hold during the execution of the operational model:
    \begin{enumerate}
        \item \label{inv:rob:ppo} If $I_1 \RobNTPreservePO I_2$ and $I_2.\DoneTS \neq \top$, then $I_1.\DoneTS \neq \top$ and $I_1.\DoneTS < I_2.\DoneTS$.
        \item \label{inv:rob:addr} If $I_1 \RobAddrDep I_2$ and $I_2.\AddrTS \neq \top$, then $I_1.\DoneTS \neq \top$ and $I_1.\DoneTS < I_2.\AddrTS$.
        \item \label{inv:rob:st-data} If $I_1 \RobDataDep I_2$, and not $I_1\RobAddrDep I_2$, and $I_2$ is a store, and $I_2.\StDataTS \neq \top$, then $I_1.\DoneTS \neq \top$ and $I_1.\DoneTS < I_2.\StDataTS$.
        \item \label{inv:rob:addr-st} If $I_1 \RobProgOrd I_2$, and $I_1$ is a memory instruction, and $I_2$ is a store, and $I_2.\DoneTS \neq \top$, then $I_1.\AddrTS \neq \top$ and $I_1.\AddrTS < I_2.\DoneTS$.
        \item \label{inv:rob:kill-done-st} We never kill a done store.
        \item \label{inv:rob:mem-val} For any address $a$, let $S$ be the store with the maximum $\DoneTS$ among all the done stores for address $a$.
        The monolithic memory value for $a$ is equal to $S.\StData$.
        \item \label{inv:rob:rf} For any done load $L$, let $S = L.\FromSt$ (i.e., $S$ is the store read by $L$).
        All of the following properties are satisfied:
        \begin{enumerate}
            \item \label{prop:rob:rf:no-kill} $S$ still exists in an ROB (i.e., S is not killed).
            \item \label{prop:rob:rf:addr-data} $S.\MemAddr = L.\MemAddr$ and $S.\StData = L.\LdVal$.
            \item \label{prop:rob:rf:done-no-st} If $S$ is done, then there is no not-done store $S'$ such that $S'.addr = a$ and $S'\RobProgOrd L$.
            \item \label{prop:rob:rf:done-max} If $S$ is done, then for any other done store $S'$ with $S'.\MemAddr = L.\MemAddr$, if $S'\RobProgOrd L$ or $S'.\DoneTS < L.\DoneTS$, then $S'.\DoneTS < S.\DoneTS$.
            \item \label{prop:rob:rf:not-done} If $S$ is not done, then $S\RobProgOrd L$, and there is no store $S'$ such that $S'.\MemAddr = L.\MemAddr$ and $S\RobProgOrd S'\RobProgOrd L$.
        \end{enumerate}
    \end{enumerate}
\end{lemma}

With the above invariants, we can finally prove the following soundness theorem.
\begin{theorem}
    GAM operational model $\subseteq$ GAM axiomatic model.
\end{theorem}
\begin{proof}
    For any execution of the operational model, at the end of the execution, all instructions must be done.
    We construct $\langle \ProgOrd, \MemOrd, \ReadFrom \rangle$ using the ending state of the operational execution as follows:
    \begin{itemize}
        \item $\ProgOrd$ is constructed as the order of instructions in each ROB.
        \item $\MemOrd$ is constructed by the ordering of $\DoneTS$, i.e., for two memory instructions $I_1$ and $I_2$, $I_1\MemOrd I_2$ iff $I_1.\DoneTS < I_2.\DoneTS$.
        \item $\ReadFrom$ is constructed by the $\FromSt$ fields, i.e., for a load $L$ and a store $S$, $S\ReadFrom L$ iff $S = L.\FromSt$.
    \end{itemize}
    Invariant~\ref{prop:rob:rf:addr-data} ensures that the constructed $\ReadFrom$ and $\ProgOrd$ are consistent with each other (e.g., it rules out the case that $\ReadFrom$ says a load should read a store with value 1, but $\ProgOrd$ says the load has value 2).
    
    Since all instructions are done at the end of execution, then invariant~\ref{prop:rob:rf:done-max} becomes the Load-Value axiom.
    Therefore, the constructed $\langle \ProgOrd, \MemOrd, \ReadFrom \rangle$ satisfy the Load-Value axiom.
    
    At the end of execution, invariant~\ref{inv:rob:ppo} becomes: if $I_1\RobNTPreservePO I_2$, then $I_1.\DoneTS < I_2.\DoneTS$.
    Note that the $\PreservePO$ computed from $\ProgOrd$ is actually the transitive closure of $\RobNTPreservePO$.
    Since instructions are totally ordered by $\DoneTS$ fields, we have if  $I_1\PreservePO I_2$, then $I_1.\DoneTS < I_2.\DoneTS$.
    Since $\MemOrd$ is defined by the order of $\DoneTS$ fields, the Inst-Order axiom is also satisfied.
\end{proof}

\subsection{Completeness: GAM Axiomatic model $\subseteq$ GAM Operational Model}\label{sec:axiom<op}
\begin{theorem}
    GAM axiomatic model $\subseteq$ GAM operational model.
\end{theorem}
\begin{proof}
    The goal is that for any legal axiomatic relations $\langle \ProgOrd, \MemOrd, \ReadFrom\rangle$ (which satisfy the GAM axioms), we can run the operational model to give the same program behavior.
    The strategy to run the operational model consists of two major phases.
    In the first phase, we only fire Fetch rules to fetch all instructions into all ROBs according to $\ProgOrd$.
    During the second phase, in each step we fire a rule that either marks an instruction as done or computes the address or data of a memory instruction.
    Which rule to fire in a step depends on the current state of the operational model and $\MemOrd$.
    Here we give the detailed algorithm that determines which rule to fire in each step:
    \begin{enumerate}
        \item If in the operational model there is a not-done reg-to-reg or branch instruction whose source registers are all ready, then we fire an Execute-Reg-to-Reg or Execute-Branch rule to execute that instruction.
        \item If the above case does not apply, and in the operational model there is a memory instruction, whose address is not computed but the source registers for the address computation are all ready, then we fire a Compute-Mem-Addr rule to compute the address of that instruction.
        \item If neither of the above cases applies, and in the operational model there is a store instruction, whose store data is not computed but the source registers for the data computation are all ready, then we fire a Compute-Store-Data rule to compute the store data of that instruction.
        \item If none of the above cases applies, and in the operational model there is a fence instruction and the guard of the Execute-Fence rule for this fence is ready, then we fire the Execute-Fence rule to execute that fence.
        \item \label{sim:mem} If none of the above cases applies, then we find the oldest instruction in $\MemOrd$, which is not-done in the operational model, and we fire an Execute-Load or Execute-Store rule to execute that instruction.
    \end{enumerate}
    Before giving the invariants, we give a definition related  to the ordering of stores for the same address.
    For each address $a$, all stores for $a$ are totally ordered by $\MemOrd$, and we refer to this total order of stores for $a$ as $<_{co}^a$.
    
    Now we show the invariants.
    After each step, we maintain the following invariants:
    \begin{enumerate}
        \item The order of instructions in each ROB in the operational model is the same as the $\ProgOrd$ of that processor in the axiomatic relations.
        \item \label{inv:result} The results of all the instructions that have been marked as done so far in the operational model are the same as those in the axiomatic relations.
        \item All the load/store addresses that have been computed so far in the operational model are the same as those in the axiomatic relations.
        \item All the store data that have been computed so far in the operational model are the same as those in the axiomatic relations.
        \item \label{inv:no-kill} No kill has ever happened in the operational model.
        \item For the rule fired in each step that we have performed so far, the guard of the rule is satisfied the at that step (i.e., the rule can fire).
        \item \label{inv:done} In each step that we have performed so far, if we fire a rule to execute an instruction (especially a load) in that step, the instruction must be marked as done by the rule.
        \item \label{inv:store} For each address $a$, the order of all the store updates on monolithic memory address $a$ that have happened so far in the operational model is a prefix of $<_{co}^a$.
    \end{enumerate}
    The detailed proof of the invariants can be found in Appendix~\ref{sec:gam_op_contain_axi}.
\end{proof}

\section{COM: an Alternative Axiomatic Model}\label{sec:COM}

In this section, we present an alternative (but still parameterized) axiomatic formulation that is perhaps less intuitive, but nevertheless in common use due to its computational efficiency.
We call this formulation the \emph{COM} model (where ``COM'' stands for communication, as described below).
We first present a proof of equivalence between the GAM axioms and the COM axioms.  This in turn implies that COM is also equivalent to the operational definition of GAM.  We then implement both axiomatic models in Alloy~\cite{alloy} in order to perform sanity checking and empirical testing of the models and of the proofs.

\subsection{The COM Axioms}
The COM model is defined in terms of three basic relations and three derived relations, plus $\PreservePO$:
\begin{itemize}
  \item Basic relations:
    \begin{itemize}
      \item Program order ($\ProgOrd$), as before
      \item Reads-from ($\ReadFrom$), as before
      \item Coherence ($\Coh$), a total order over the writes to each memory address
    \end{itemize}
  \item Derived relations:
    \begin{itemize}
      \item Reads-from external ($\Rfe$), which is the subset of $\ReadFrom$ for which both the read and the write are in different threads
      \item From-reads ($\Fr$=${\RfInv};\Coh$), which relates each read $r$ to every write which follows the $\ReadFrom$-source of $r$ in $\Coh$.  ($\RfInv$ indicates the inverse of $\ReadFrom$)
      \item Program order, same location ($\PoLoc$), which is the subset of program order that relates memory accesses to the same memory address
    \end{itemize}
\end{itemize}

Another derived relation $<_{com}=\ReadFrom \cup \Coh \cup \Fr$ is often defined as a convenient shorthand in this style of model (hence our choice of the name ``COM''), but we do not use it in this paper.

\noindent In the COM formulation, an execution is legal if it satisfies the following two axioms:

\begin{itemize}
  \item \textbf{Axiom SC-per-Location:} $\acyclic(\ReadFrom \cup \Coh \cup \Fr \cup \PoLoc)$
  \item \textbf{Axiom Causality:} $\acyclic(\Rfe \cup \Coh \cup \Fr \cup \PreservePO)$
\end{itemize}

\subsection{Equivalence of GAM and COM}

The complete proofs are provided in Appendix~\ref{sec:gamcom}.  We provide an intuition here.

To prove that GAM $\subseteq$ COM, we must do two things: 1) find a suitable choice of $\Coh$, which does not exist in the GAM model, and 2) prove that if the GAM axioms are satisfied, the COM axioms are satisfied.  Of course, the natural choice for $\Coh$ is to simply take the restriction of $\MemOrd$ that relates only stores to the same address, and that is indeed what we use.  It remains to show that for any choice of $\MemOrd$ in the GAM axioms, the two COM axioms are satisfied.

We start with a lemma:
\begin{lemma}\label{lem:com_in_memord}
  All of $\Rfe$, $\Coh$, $\Fr$, and $\PreservePO$ are contained in $\MemOrd$.
\end{lemma}
\begin{proof}
  Straightforward; see appendix.
\end{proof}

With this lemma, it is easy to show that the Causality axiom is satisfied:
\begin{theorem}
  The Causality axiom is satisfied.
\end{theorem}
\begin{proof}
  By Lemma~\ref{lem:com_in_memord}, the union $\Rfe \cup \Coh \cup \Fr \cup \PreservePO$ is a subset of $\MemOrd$.  Therefore, since $\MemOrd$ is acyclic, $\Rfe \cup \Coh \cup \Fr \cup \PreservePO$ must also be acyclic.
\end{proof}

The SC-per-Location axiom will take a bit more work to prove.
To start, define $\Eco$ as the union of the following relations:
\begin{itemize}
  \item $\Coh$ (Write to Write)
  \item $\Fr$ (Read to Write)
  \item ${\Coh}^*;\ReadFrom$ (Write to Read)
  \item $\RfInv;{\Coh}^*;\ReadFrom$ (Read to Read)
\end{itemize}

\begin{lemma}\label{lem:eco_either}
  For all pairs $i_1$, $i_2$ of memory accesses to the same address, either $i_1\Eco i_2$ or $i_2\Eco i_1$.
\end{lemma}
\begin{proof}
  By construction; see appendix.
\end{proof}

If $i_1$ and $i_2$ are related in program order, then the $\Eco$ direction must match:
\begin{lemma}\label{lem:eco_poloc}
  If $i_1\PoLoc i_2$, then $i_1\Eco i_2$.
\end{lemma}
\begin{proof}
  The alternative of $i_2\Eco i_1$ results in a contradiction, except for one case where it overlaps $i_1\Eco i_2$.  See appendix.
\end{proof}

\begin{theorem}
  The SC-per-Location axiom is satisfied.
\end{theorem}
\begin{proof}
  (abbreviated; see appendix)

  First, by Lemma~\ref{lem:eco_poloc}, all $\PoLoc$ edges involving at least one write can be converted into sequences containing only $\ReadFrom$, $\Coh$, and $\Fr$.  So we consider only cycles with $\ReadFrom$, $\Coh$, $\Fr$, and read-to-read $\PoLoc$ edges.  Replace every instance of read-read $\PoLoc$ in the cycle with $\RfInv;{\Coh}^*;\ReadFrom$ per Lemma~\ref{lem:eco_poloc}.  Now, because $\Coh$ and $\Fr$ both target writes, every appearance of $\RfInv$ must be preceded either by $\ReadFrom$ or by $\RfInv;{\Coh}^*;\ReadFrom$.  In particular, every appearance of $\RfInv$ must be preceded directly by $\ReadFrom$.  Since $\ReadFrom;\RfInv$ is the identity function, all appearances of $\RfInv$ in the cycle can be eliminated by simply removing each $\ReadFrom;\RfInv$ pair in the cycle.  This leaves a cycle with only $\ReadFrom$, $\Coh$, and $\Fr$, which is a contradiction.
\end{proof}

\subsection{COM $\subseteq$ GAM}

This direction is easier.  Given $\ProgOrd$, $\ReadFrom$, and $\Coh$, we must find a suitable $\MemOrd$.
By the Causality axiom, $\Rfe \cup \Coh \cup \Fr \cup \PreservePO$ is acyclic, and hence there is at least one total ordering compatible with it.  We show that any such total ordering satisfies GAM.
The Inst-Order axiom is true by construction, and hence we must only show that the Load-Value axiom is satisfied.

\begin{theorem}
  Any $\MemOrd$ which is a total ordering of $\Rfe \cup \Coh \cup \Fr \cup \PreservePO$ satisfies the Load-Value axiom.
\end{theorem}
\begin{proof}
  If $w\ReadFrom r$, then either $w\Rfi r$ or $w\Rfe r$.  In the first case, $w\ProgOrd r$, or else it would contradict the SC-per-Location axiom.  In the second case, $w\MemOrd r$ by construction of $\MemOrd$.  In either case, $w$ must be in the candidate set
  \[ \{
    \mathsf{St}\ a\ v'\ |\ \mathsf{St}\ a\ v' \ProgOrd
    \mathsf{Ld}\ a\ \vee\ \mathsf{St}\ a\ v' \MemOrd \mathsf{Ld}\ a \}.
  \]
It remains to be shown that $w$ is in fact the $\MemOrd$-maximal element of that candidate set.

  Suppose that $w$ is not maximal.  Then there is some other write $w'$ to the same address $a$ such that $w\MemOrd w'$ and either $w'\ProgOrd r$ or $w'\MemOrd r$.  But then by definition, $r\Fr w'$, and $\Fr$ cannot contradict either $\ProgOrd$ (by SC-per-Location) or $\MemOrd$ (by construction of $\MemOrd$).  Hence we have a contradiction.
\end{proof}

\subsection{Empirical Validation}

We also used model checking to confirm the validity of the proof of equivalence between GAM and COM.
We encoded both models into Alloy~\cite{alloy,memalloy}, a relational model finder backed by a SAT solver, and checked for any mismatches.
The definition of this model is shown in Appendix~\ref{sec:alloy}.
In keeping with the spirit of the proofs, $\PreservePO$ is entirely parameterized; there is no explicit notion of fence or dependency in this version of the model.  We only assume that Definition~\ref{def:ppo-same-addr} always holds.
Under these conditions, Alloy verifies in roughly one hour that no counterexamples are found for tests with up to seven instructions.

\section{Comparing GAM with Existing Atomic Memory Models}\label{sec:instance}

Now that we have defined our three model formulations and completed the proofs of equivalence, we can now show how GAM is related to existing atomic memory models.
Most atomic memory models already have the same axioms as GAM, so our commparison will base off from the definitions of $\PreservePO$.
In some cases, the existing memory model can be instantiated from GAM.
While in other cases, the dependency ordering or same-address load-load ordering of an existing model does not match that in GAM, and we will explain the difference and possible ways to tweak GAM to match the existing model.

\subsection{SC}
SC has no fence, the memory/fence ordering enforced by SC is shown in Table~\ref{tab:order-sc}.
\begin{table}[!htb]
    \centering
    \begin{tabular}{|l|l|l|}
        \hline
        \backslashbox{$I_{old}$}{$I_{new}$} & Ld   & St   \\ \hline
        Ld                                  & True & True \\ \hline
        St                                  & True & True \\ \hline
    \end{tabular}
    \caption{Memory/fence orderings for SC: $\OrderedFunc_{SC}(I_{old}, I_{new})$}\label{tab:order-sc}
\end{table}

After supplying $\OrderedFunc_{SC}$ to GAM, $\PreservePO$ in the GAM axiomatic instance will order every pair of memory instructions from the same processor.
In this case, the Load-Value axiom will reduce to
\begin{displaymath}
\mathsf{St}\ a\ v \ReadFrom \mathsf{Ld}\ a \Rightarrow
\mathsf{St}\ a\ v = \max_{mo}\{
\mathsf{St}\ a\ v' \MemOrd \mathsf{Ld}\ a \}
\end{displaymath}
This is because $I_1\ProgOrd I_2$ implies $I_1\PreservePO I_2\Rightarrow I_1\MemOrd I_2$.
Thus, the GAM axiomatic instance is equivalent to SC.
In the operational instance of GAM, the SC-$\OrderedFunc$ function makes the guards of Execute-Load and Execute-Store rules to wait for all previous memory instructions to be done.
This is also the same as SC.

\subsection{TSO}
TSO has only one fence, and the memory/fence ordering enforced by TSO is shown in Table~\ref{tab:order-tso}.
The TSO axiomatic model defines a $\PreservePO$ edge from instructions $I_1$ to $I_2$ iff $\OrderedFunc_{TSO}(I_1, I_2)$.
When supplying GAM with the $\OrderedFunc_{TSO}$ function, the $\PreservePO$ of the resulting GAM axiomatic instance is the same as that of TSO axiomatic model, since $\DepOrd$ and $\SameAddrOrd$ are entirely contained within $\OrderedFunc_{TSO}$.
In other words, all load-load, load-store, and store-store orderings are automatically enforced anyway, so there is no need to worry about any particular subset of such orderings.
In the operational instance of GAM, the TSO-$\OrderedFunc$ function will cause the guard of the Execute-Load rule to wait for all older loads in ROB to be done, and cause the guard of the Execute-Store rule to wait for all older memory instructions to be done.

\subsection{SPARC RMO}\label{sec:compare-rmo}
RMO has various fences, and the memory/fence orderings enforced by RMO are shown in Table~\ref{tab:order-rmo}.
RMO also enforces the ordering between dependent instructions.
However, there is a bug in the dependency definition in RMO~\cite{wmm}.
For the sake of comparison, we consider this to be a mistake rather than an intentional deviation, and hence we simply assume a corrected version of RMO that has the same definition of dependency ordering as GAM does.

When we supply the $\OrderedFunc_{RMO}$ function to GAM, the resulting GAM axiomatic instance is very close to but slightly different from the RMO axiomatic model.
The difference is that RMO does not order loads for the same address.
Same-address load-load reordering is a subtle issue (see Section~\ref{sec:compare-arm}) and a common source of implementation bugs~\cite{arm:llh}, but by modern standards RMO's approach is considered overly aggressive.
Nevertheless, for completeness we describe how GAM could be tweaked to allow same-address load-load reordering: we can simply tweak the axiomatic definition of GAM by removing case \ref{ppo:ld->ld} from the definition of $\SameAddrOrd$ (Definition~\ref{def:ppo-same-addr}).
After this removal, the GAM axiomatic instance becomes exactly the same as RMO.

The challenge is then to tweak the GAM operational instance to keep it equivalent to the axiomatic instance.
In the GAM operational instance, we relax the Execute-Load rule by making the ROB search ignore loads for the same address.
Also, in the Compute-Mem-Addr rule that computes the address of a \emph{load}, the ROB search in the rule should ignore younger loads for the same address.
These two changes relax the ordering between loads for the same address, making the operational instance of GAM still match the axiomatic instance of GAM.

\subsection{WMM}

WMM~\cite{wmm} has two fences: Commit and Reconcile, and the memory/fence orderings enforced by WMM are shown in Table~\ref{tab:order-wmm}.
The ordering between Commit and Reconcile is particularly important in WMM, as preventing store-load reordering requires the combination of a Commit and a Reconcile.

\begin{table}[!htb]
    \centering
    \begin{tabular}{|l|l|l|l|l|}
        \hline
        \backslashbox{$I_{old}$}{$I_{new}$} & Ld    & St    & Commit & Reconcile \\ \hline
        Ld                                  & False & True  & True   & True \\ \hline
        St                                  & False & False & True   & False \\ \hline
        Commit                              & False & True  & True   & True \\ \hline
        Reconcile                           & True  & True  & True   & True \\ \hline
    \end{tabular}
    \caption{Memory/fence orderings for WMM: $\OrderedFunc(I_{old}, I_{new})$}\label{tab:order-wmm}
\end{table}

WMM does not enforce any dependency ordering (although all load-store ordering is automatically enforced).
Therefore, in order to make GAM match WMM, we must first tweak the axiomatic definition of GAM by dropping $\DepOrd$.
After this change, one subtle difference still remains.
Consider a scenario in which $L_1 \ProgOrd S \ProgOrd L_2$, where $L_1$ and $L_2$ are both loads for address $a$ and $S$ is store for $a$.
GAM does not directly require $L_1$ and $L_2$ to be ordered in $\MemOrd$ due to the intervening store, but WMM does require $L1 \MemOrd L_2$.
However, it turns out that the two are actually equivalent in this case, because we can transform the $\MemOrd$ in GAM to a legal $\MemOrd$ in WMM (i.e., one that obeys $\PreservePO$ in WMM).

During the transformation, the store read by each load determined by the Load-Value axiom will not change.
In each transformation step, for a processor $i$, we find one such potential counterexample scenario: a load $L_2$ which is the youngest load in the $\ProgOrd$ of processor $i$ which is $\MemOrd$-before an older load $L_1$ from processor $i$ for the same address.
This gives us $L_1\ProgOrd L_2$ and $L_2\MemOrd L_1$.
Since this reordering is allowed by GAM, there must be store for $a$ between $L_1$ and $L_2$ in the $\ProgOrd$ of processor $i$.
The transformation is to move $L_2$ to be right after $L_1$ in $\MemOrd$.
This is legal because no WMM ordering primitive can cause an instruction to be ordered after $L_2$ without also being ordered after $L_1$.
It also does not affect the value returned by $L_1$, nor does it affect any instruction originally older than $L_2$ in $\MemOrd$.  
By repeating the above steps, we can complete the transformation until no such apparent contradictions remain.
Therefore, the axiomatic instance of GAM is equivalent to WMM.

We also need to tweak the GAM operational instance to keep it equivalent to the axiomatic instance.
In the operational instance of GAM, if a Fetch rule fetches a load into the ROB, we predict the load value and record it in the new ROB entry.
Younger instructions in ROB can read the predicted load value for computation.
In the Execute-Load rule, if the load is marked as done, then we compare the read value with the previously predicted value.
In case they are not equal, we kill all instructions younger than the load in ROB.
Introducing load-value prediction relaxes dependency ordering, making the operational instance of GAM still match the axiomatic instance of GAM.

\subsection{ARM v8.2}\label{sec:compare-arm}

As of March 2017, ARM completely revamped its memory consistency model.
The end result looks very similar to GAM, with $\FenceOrd$ defined to include \texttt{DMB LD} (load-to-load/store), \texttt{DMB ST} (store-to-store), Load-Acquire (ordered with subsequent loads/stores), and Store-Release (ordered with prior loads/stores).
There is, however, one main exception: ARM allows read-same-write (RSW) behavior (Figure~\ref{fig:rsw}): two loads which return the value written by the same write are \emph{not} ordered in $\PreservePO$.
In particular, in Figure~\ref{fig:rsw}, the two loads of $z$ are not ordered on ARM, even though they are two loads of the same address with no intervening store.
If the two loads read from \emph{different} stores (e.g., the RDW behavior in Figure~\ref{fig:rsw}), however, the outcome is forbidden.

\begin{figure}[!htb]
    \centering
  \begin{tabular}{l||l}
    St $x$, $1$ & Ld r1, $y$ (=1)       \\
    Fence       & Ld r2, $z+r1-r1$ (=0) \\
    St $y$, $1$ & Ld r3, $z$ (=0)       \\
                & Ld r4, $x+r3-r3$ (=0) \\
  \end{tabular}
  \hspace{1cm}
  \begin{tabular}{l||l}
    St $x$, $1$ & Ld r1, $y$ (=1)       \\
    Fence       & Ld r2, $z+r1-r1$ (=1) \\
    St $y$, $1$ & Ld r3, $z$ (=2)       \\
                & Ld r4, $x+r3-r3$ (=0) \\
  \end{tabular}
  \caption{The read-same-writes (RSW, left) litmus test is forbidden under GAM but permitted by ARM.
      The read-different-writes (RDW, right) litmus test is forbidden under both ARM and GAM.
      Both tests are the same, but ARM makes a distinction based on the values returned by the loads.}
  \label{fig:rsw}
\end{figure}

We feel the subtlety in allowing the RSW behavior while forbidding the RDW behavior may lead to confusion.
Besides, there is no published evidence showing that having this subtlety can lead to higher performance in implementations.
Therefore, definition~\ref{def:ppo-same-addr}.\ref{ppo:ld->ld} of GAM was carefully chosen to forbid the RSW behavior, while still allowing so-called ``fri-rfi'' behavior (Figure~\ref{fig:fri_rfi}) which can result from local store forwarding in implementations.

\begin{figure}[!htb]
  \centering
  \begin{tabular}{l||l}
    St $x$, $1$ & Ld r1, $y$ (=1)       \\
    Fence       & St $y$, $2$           \\
    St $y$, $1$ & Ld r2, $y$ (=2)       \\
                & Ld r3, $x+r2-r2$ (=0) \\
  \end{tabular}
  \caption{The MP+fence+fri-rfi-addr litmus test.}
  \label{fig:fri_rfi}
\end{figure}

\subsection{Alpha}

Alpha's memory model is similar to GAM with one single fence, but it is strictly weaker in that it
does not enforce any dependencies, including even load-store dependencies.  Alpha therefore allows the behavior in Figure \ref{fig:alpha}, while GAM does not.

\begin{figure}[t]
    \centering
  \begin{tabular}{l||l}
    Ld r1, $x$ (=1) & Ld r2, $y$ (=1)\\
    If r1 == 0      & St $x$, r2\\
    then St $y$, 1  &\\
    else St $y$, 1  &\\
  \end{tabular}
  \caption{Alpha is more relaxed to reorder stores before branches}
  \label{fig:alpha}
\end{figure}

It is possible to remove all dependency orderings from the GAM axiomatic model in order to account for this behavior, but doing so in the operational model would be a substantial challenge (just as it would be in any real microarchitecture).  It is not generally possible to perform speculative stores, as there is no way to undo a failed speculation, so the operational model that produces such behaviors would necessarily be somewhat contrived.
In any case, such behaviors are no longer produced in more modern memory model definitions, and so we do not attempt to adapt the GAM operational model to account for speculative load-store dependency reordering.

\subsection{RISC-V}
The RISC-V model is not yet finalized, but it is likely to use a model very similar to GAM.
For comparison, we include the basics of the expected model below.
Note in particular that Release is not ordered with Acquire, in contrast to how WMM does order Commit with Reconcile.

\begin{figure}[!htb]
    \centering
    \begin{tabular}{|l|l|l|l|l|l|}
        \hline
        \backslashbox{$I_{old}$}{$I_{new}$} & Ld    & St    & Release & Acquire & Full \\ \hline
        Ld                                  & False & False & True    & True    & True \\ \hline
        St                                  & False & False & True    & False   & True \\ \hline
        Release                             & False & True  & True    & False   & True \\ \hline
        Acquire                             & True  & True  & True    & True    & True \\ \hline
        Full                                & True  & True  & True    & True    & True \\ \hline
    \end{tabular}
    \caption{Memory/fence orderings for RISC-V: $\OrderedFunc_{RISC\Hyphen{}V}(I_{old}, I_{new})$}
    \label{tab:riscv}
\end{figure}

This model presents all of the best features of GAM: a minimal set of dependency orderings that are nevertheless up to modern standards, a flexible and performant yet easy-to-define set of fences, and a reasonably-minimal set of same-address orderings.
It can also be adapted to the needs of any subtle variant or modification by simply changing the set of fences that are included in Table~\ref{tab:riscv}.  As such, if RISC-V adopts GAM, it will be the first modern architecture allowing load-store reordering to come complete with a proper axiomatic model, a proper operational model, and a full proof of equivalence.

\section{GAM-I2E: Parameterizing Dependency Ordering}\label{sec:I2E}
\label{sec:i2e}

In previous sections, we have seen that GAM is not parameterized by dependency orderings, and requires manual tweak on the definitions to produce memory models with a different dependency ordering.
The major reason is that GAM is designed to be able to allow load-store reordering.
As stated in Section~\ref{sec:inst-reorder}, allowing load-store reordering means a store may indirectly affect an older load in the same processor.
This implies that no matter what mechanism an operational model uses, it cannot execute instructions in order.
Thus, when an operational model wants to execute an instruction $I$, it may not have all the information (e.g., memory access addresses) of instructions that are older than $I$ in the same processor.
However, in the axiomatic model, such information is always available, and will used in the computation  of $\DepOrd$ edges that point to $I$.
The lack of information in the operational model makes it difficult to parameterize dependency ordering while keeping the operational and axiomatic models equivalent.

Recently, Zhang et al. have shown that some memory models can be expressed in the form of instantaneous instruction execution (I2E) when load-store reordering is forbidden.
I2E means that each processor in the operational model executes instructions instantaneously and in order.
Here we apply that idea, i.e., we force $\OrderedFunc$(Ld, St) to be true (i.e., forbid load-store reordering) to make it possible to express the GAM operational model in I2E.
In the I2E operational model, for the next instruction $I$ to execute on a processor $i$, we know all the information of all the instructions older than $I$ in processor $i$.
Thus, the I2E operational model can compute the $\PreservePO$ edges pointing to $I$ using the same way as the axiomatic model does.
Hence, the I2E operational model knows the same constraint of executing $I$ (i.e., the constraint on placing $I$ in the global order) as the axiomatic model does.
I2E eliminates the difference in information available to the axiomatic and operational models, making it possible to parameterize the memory model by any form of $\DepOrd$.
The model is not parametrized by same-address ordering because some same-address ordering are required by single-thread correctness.
It should be noted that computing the $\PreservePO$ edges pointing to instruction $I$ should not require knowing the execution result of $I$.
This is true for computing the $\LSFOrd$ and $\SameAddrOrd$ edges defined in GAM.
This should also be true for most definitions of preserved dependency ordering (i.e., $\DepOrd$).

We refer to this new model as \emph{GAM-I2E}.
In the following, we give the axiomatic and operational definitions of GAM-I2E, which are parametrized by $\LSFOrd$ and $\DepOrd$, as well as the equivalence proof.

\subsection{Axiomatic Model of GAM-I2E}

The axiomatic model of GAM-I2E is exactly the same as that of GAM in Section~\ref{sec:gam:axiom}.
The only additional requirement is that $\OrderedFunc$(Ld, St) must be true.
Given this requirement, one can slightly simplify the definition of $\SameAddrOrd$ by removing case~\ref{ppo:ld->st} from Definition~\ref{def:ppo-same-addr}, because that load-store ordering is already enforced by $\LSFOrd$.

\subsection{Operational Model of GAM-I2E}
The operational model consists of $n$ processors.
Each processor executes instructions instantaneously, and contains a local buffer to temporarily keep executed stores and fences.
The memory system is a list $\IIEMemOrd$ of load and store instructions. (we use suffix $i2e$ to distinguish from the definitions in the axiomatic model).
In the following, we will also use $\IIEMemOrd$ as a total order of memory instructions in the memory system, i.e., $I_1 \IIEMemOrd I_2$ means that instruction $I_1$ is closer to the list head than $I_2$.

Assume the next instruction to execute on a processor is $I$.
Let $\IIEProgOrd$ be the execution order of $I$ and all instructions already executed by the processor (i.e., $I$ is the youngest).
If we treat $\IIEProgOrd$ as a program order, then we can follow the definitions of preserved program order (Section~\ref{sec:ppo}) to compute $\IIEPreservePO$ from $\IIEProgOrd$.
$\IIEPreservePO$ is the preserved program order among $I$ and all instructions executed by the processor.
Note that to make this definition meaningful, computing $\IIEPreservePO$ should not require knowing the load value of $I$.
It should also be noted that $\IIEPreservePO$ edges grow monotonously.
To be specific, consider the case that a processor has executed $k$ instructions $I_1 \IIEProgOrd I_2 \IIEProgOrd \cdots \IIEProgOrd I_k$, and we have computed the $\IIEPreservePO$ edges for $I_1\ldots I_k$.
If the processor executes a new instruction $I_{k+1}$, then the $\IIEPreservePO$ edges for $I_1\ldots I_{k+1}$ will contain all the previously computed $\IIEPreservePO$ edges for $I_1\ldots I_k$, and all the newly added edges will point to $I_{k+1}$.
This is because whether instructions $I$ and $I'$ are ordered by preserved program order is fully determined by $I$, $I'$ and instructions between $I$ and $I'$ in the program order.

With the above definitions, now we give the rules for the operational moddel of GAM-I2E.
\begin{itemize}
    \item \textbf{Rule Execute-Reg-Branch:} Execute a reg-to-reg or branch instruction $I$. \\
    \emph{Guard:} True. \\
    \emph{Action:} Execute $I$ and update local register states.
    
    \item \textbf{Rule Execute-Store-Fence:} Execute a store or fence instruction $I$. \\
    \emph{Guard:} True. \\
    \emph{Action:} Insert $I$ into the local buffer.
    
    \item \textbf{Rule Execute-Load:} Execute a load $L$ for address $a$. \\
    \emph{Guard:} There is no instruction $I$ in the local buffer that is ordered before $L$ in $\IIEPreservePO$. \\
    \emph{Action:} Insert $L$ into an arbitrary place in list $\IIEMemOrd$ such that for any memory instruction $I$ which is ordered before $L$ in $\IIEPreservePO$, $L$ is after $I$ in $\IIEMemOrd$.
    With the updated $\IIEMemOrd$, we can determine the load value of $L$ in the following way:
    \begin{enumerate}
        \item If the local buffer contains any store for $a$, then $L$ reads from the youngest (i.e., most recently inserted) store for $a$ in the local buffer.
        \item Otherwise, $L$ reads from the youngest store for $a$ in $\IIEMemOrd$ that is from the same processor of $L$ or is older than $L$ in $\IIEMemOrd$.
    \end{enumerate}
    
    \item \textbf{Rule Dequeue-Store:} Dequeue a store $S$ from the local buffer to the memory system. \\
    \emph{Guard:} There is no instruction in the local buffer that is ordered before $S$ in $\IIEPreservePO$. \\
    \emph{Action:} Remove $S$ from the local buffer, and append $S$ to the end of $\IIEMemOrd$ (i.e., $S$ becomes the youngest in $\IIEMemOrd$).
    
    \item \textbf{Rule Dequeue-Fence:} Dequeue a fence $F$ from the local buffer. \\
    \emph{Guard:} There is no instruction in the local buffer that is ordered before $F$ in $\IIEPreservePO$. \\
    \emph{Action:} Remove $F$ from the local buffer.
\end{itemize}

\subsection{Soundness: GAM-I2E Operational Model $\subseteq$ GAM-I2E Axiomatic Model}
\begin{theorem}
    GAM-I2E operational model $\subseteq$ GAM-I2E axiomatic model
\end{theorem}
\begin{proof}
The goal is to show that for any execution of the GAM-I2E operational model, we can construct $\langle \ProgOrd, \MemOrd, \ReadFrom\rangle$ which satisfies the GAM-I2E axioms and has the same program behavior as the operational execution.
$\ProgOrd$ is the order of executing instructions in each processor of the operational model.
$\ReadFrom$ is constructed according to the Execute-Load rule, i.e., if the Execute-Load rule picks store $S$ to satisfy a load $L$, then $S\ReadFrom L$.
$\MemOrd$ is the $\IIEMemOrd$ at the end of the operational execution.
We need to show that $\langle \ProgOrd, \ReadFrom, \MemOrd \rangle$ satisfies the axioms.
It should be noted that $\IIEProgOrd$ and $\IIEPreservePO$ always matche $\ProgOrd$ and $\PreservePO$ respectively during the operational execution.
That is, when an instruction $I$ of processor $i$ is executed in the operational execution, $\IIEProgOrd$ and $\IIEPreservePO$ of instructions executed by processor $i$ (including $I$) satisfies the following invariants:
\begin{itemize}
    \item $\IIEProgOrd$ is a prefix of $\ProgOrd$ (of processor $i$) up to $I$ (including $I$).
    \item For any instructions $I_1 \PreservePO I_2$ from processor $i$, if $I_1$ and $I_2$ are not ordered after $I$ in $\ProgOrd$ (i.e., $I_2$ may be equal to $I$), then $I_1 \IIEPreservePO I_2$.
    \item For any instructions $I_1$ and $I_2$, if $I_1 \IIEPreservePO I_2$, then $I_1 \PreservePO I_2$.
\end{itemize}

With above invariants, we prove that the Inst-Order axiom is satisfied by contradiction, i.e., we assume there are two memory instructions $I_1$ and $I_2$ from processor $i$ such that $I_1 \PreservePO I_2$ but $I_2 \MemOrd I_1$.
In the operational model, when $I_2$ is executed, $I_1$ must have been executed, and $I_1$ is ordered before $I_2$ in $\IIEPreservePO$ according to the invariants.
At the time when $I_2$ is executed, $I_1$ can only be in one of the following two places:
\begin{enumerate}
    \item $I_1$ is already in the memory system:
    In this case, if $I_2$ is a load, then the Execute-Load rule ensures that $I_2$ is placed after $I_1$ in $\IIEMemOrd$.
    If $I_2$ is a store, it can only be appended to the end of $\IIEMemOrd$, and is still after $I_1$ in $I\IIEMemOrd$.
    
    \item $I_1$ is in the local buffer:
    In this case, $I_1$ must be a store.
    $I_2$ must also be a store (otherwise if $I_2$ is a load, the guard of Execute-Load rule will be false due to $I_1$ in the local buffer).
    And $I_2$ is inserted into the local buffer.
    The Dequeue-Store rule ensures that $I_1$ will be appended to $\IIEMemOrd$ before $I_2$, so $I_1$ is still before $I_2$ in $\IIEMemOrd$.
\end{enumerate}
$I_1\IIEMemOrd I_2$ implies that $I_1\MemOrd I_2$, contradicting with the initial assumption.
Thus the Inst-Order axiom is satisfied.

Now we show that the Load-Value axiom is also satisfied.
Consider a load $L$ for address $a$ from processor $i$ which reads from a store $S$ in the operational execution.
When the Execute-Load rule executes $L$, we consider where $S$ resides:
\begin{enumerate}
    \item $S$ is in the local buffer of processor $i$:
    $S$ will be appended to $\IIEMemOrd$ later, so $S$ must be after $L$ in $\MemOrd$.
    Since we already have $S \ProgOrd L$ and $L\MemOrd S$, the Load-Value axiom will only pick stores that are before $L$ in $\ProgOrd$.
    Now we consider such a store $S'$ ($\neq S$) for $a$ which is before $L$ in $\ProgOrd$.
    Note that $S$ is the most recently inserted store for $a$ when $L$ is executed.
    Thus, when $S$ is executed by processor $i$, $S'$ must have been executed, and we have $S' \IIEProgOrd S \Rightarrow S' \IIEPreservePO S$ at that time (according the definition of same-address ordering).
    Therefore, $S'$ is appended to $\IIEMemOrd$ before $S$, and thus $S'\MemOrd S$.
    As a result, the Load-Value axiom also agrees on $S\ReadFrom L$.
    
    \item $S$ is already in $\IIEMemOrd$:
    When $L$ is executed, the local buffer of processor $i$ cannot contain any store for $a$ according to the guard of the Execute-Load rule.
    Thus, all stores for $a$ that are before $L$ in $\ProgOrd$ are already in $\IIEMemOrd$ at that time.
    Since stores can only be appended to the end of $\IIEMemOrd$, all stores for $a$ that are before $L$ in $\MemOrd$ are also in $\IIEMemOrd$ by the time when $L$ is executed.
    Then the way that the Execute-Load rule determines the load value of $L$ is exactly the same as the Load-Value axiom.
\end{enumerate}
\end{proof}

\subsection{Completeness: GAM-I2E Axiomatic Model $\subseteq$ GAM-I2E Operational Model}
\begin{theorem}
    GAM-I2E axiomatic model $\subseteq$ GAM-I2E operational model.
\end{theorem}
\begin{proof}
The goal is that for any legal axiomatic relations $\langle \ProgOrd, \MemOrd, \ReadFrom\rangle$ (which satisfy the GAM-I2E axioms), we can run the GAM-I2E operational model to simulate the same program behavior.
In each step of the simulation, we first decide which rule to fire in the operational model based on the current state of the operational model and $\MemOrd$, and then we fire that rule.
Here is the algorithm to determine which rule to fire in each simulation step:
\begin{enumerate}
    \item If in the operational model there is a processor whose next instruction is not a load, we fire an Execute-Reg-Branch or Execute-Store-Fence rule to execute that instruction in the operational model.
    \item If the above case does not apply, and in the operational model there is a fence that can be dequeued from the local buffer, then we fire the Dequeue-Fence rule to dequeue that fence in the operational model.
    \item \label{sim:i2e:st} If neither of the above cases applies, and in the operational model there is a store $S$ in the local buffer of a processor, and $S$ can be dequeued from the local buffer (i.e., the guard for the Dequeue-Store rule is true), and all stores before $S$ in $\MemOrd$ are already in $\IIEMemOrd$, then we fire a Dequeue-Store rule to dequeue $S$ in the operational model.
    \item \label{sim:i2e:ld} If none of the above cases applies, then in the operational model there must be a processor such that the next instruction of the processor is a load $L$, and $L$ can be executed (i.e., the guard for the Execute-Load rule is true), and all stores before $L$ in $\MemOrd$ are already in $\IIEMemOrd$.
    We fire an Execute-Load rule to execute $L$ in the operational model.
    In the Execute-Load rule of $L$, we insert $L$ into $\IIEMemOrd$ such that for any instruction $I$ already in $\IIEMemOrd$, if $I \MemOrd L$ then $I \IIEMemOrd L$, otherwise $L \IIEMemOrd I$.
\end{enumerate}
After each step of the simulation, we keep the following invariants:
\begin{enumerate}
    \item \label{inv:i2e:po} The execution order on each processor is a prefix of the $\ProgOrd$ of that processor.
    \item The result of each executed instruction is the same as that in $\ProgOrd$.
    \item \label{inv:i2e:rf} The store read by each executed load is the same as that indicated by the $\ReadFrom$ edges.
    \item \label{inv:i2e:stuck} The simulation cannot get stuck.
    \item \label{inv:i2e:mo-match} For two memory instruction $I_1$ and $I_2$, if $I_1 \IIEMemOrd I_2$ in the operational model, then $I_1 \MemOrd I_2$ in the axiomatic relations.
    \item \label{inv:i2e:mo-prefix} The order of all stores in $\IIEMemOrd$ is a prefix of the order of all stores in $\MemOrd$.
\end{enumerate}
The first two induction invariants imply that \emph{before} each simulation step, the following properties hold for each processor $i$ (assuming the next instruction of the processor is $I$):
\begin{enumerate}
    \item $\IIEProgOrd$ is a prefix of $\ProgOrd$ (of processor $i$) up to $I$ (including $I$).
    \item For any instructions $I_1 \PreservePO I_2$ from processor $i$, if $I_1$ and $I_2$ are not ordered after $I$ in $\ProgOrd$ (i.e., $I_2$ may be equal to $I$), then $I_1 \IIEPreservePO I_2$.
    \item \label{prop:i2e:ppo} For any instructions $I_1$ and $I_2$, if $I_1 \IIEPreservePO I_2$, then $I_1 \PreservePO I_2$.
\end{enumerate}
The detailed proof for these invariants can be found in Appendix~\ref{sec:i2_op_contain_axi}.
\end{proof}

It should be noted that the above models and proofs of GAM-I2E do not rely on the specific forms of $\DepOrd$ or $\LSFOrd$.
Therefore, GAM-I2E is fully parametrized by $\DepOrd$ and $\LSFOrd$.

\section{Conclusion}\label{sec:conclusion}

For years, many of the leading industry memory models have been so complicated to understand and to analyze that the status quo was simply to live with an incomplete and underspecified memory model.
Academics would attempt to build axiomatic and operational models and then to prove them equivalent, but these models and proofs were subject to frequent breakage and refinement due to the thorniness of the issues at hand.
Other models were simply never updated to modern standards, and were therefore left with definitions fence ordering, same-address ordering, and/or dependency ordering that are today well known to be insufficient.
This has led to no shortage of confusion in the broader understanding of memory models in the field.

In response to the recently emerging trend back towards atomic memory models, we present GAM, a flexible operational \emph{and} axiomatic memory model definition that is \emph{parameterized} by the set of fences in the model.
GAM corrects the preserved program order definition oversights present in memory models from past generations, and it reduces the definition of fence behavior into localized intra-thread ordering specifications that can be easily understood in isolation.
GAM also comes with proofs of equivalence between its axiomatic and operational models, thereby overcoming the obstacle that many previous memory models have faced in being far too complicated to understand or to work with.
The equivalence makes it much easier for architects, programmers, and theoreticians to each simply use the variant that they find easiest to work with.

Finally, GAM also makes it easy to understand the implications of tweaking a memory model's definition.
It is easy to add new fences that trade off strength for performance, for example.
It is also possible to remove behaviors; as we show, forbidding load-store reordering altogether allows GAM to be reduced to an even simpler I2E-based definition.
We believe that all of these features will go a long way towards eliminating the worst of the subtleties and corner cases that have most of the memory models of past generations.

\bibliographystyle{plain}
\bibliography{ref}

\begin{thebibliography}{10}

\bibitem{riscv}
The risc-v instruction set.
\newblock \url{https://riscv.org/}.

\bibitem{alpha1998}
{\em Alpha Architecture Handbook, Version 4}.
\newblock Compaq Computer Corporation, 1998.

\bibitem{adve1996shared}
Sarita~V Adve and Kourosh Gharachorloo.
\newblock Shared memory consistency models: A tutorial.
\newblock {\em computer}, 29(12):66--76, 1996.

\bibitem{alglave2012formal}
Jade Alglave.
\newblock A formal hierarchy of weak memory models.
\newblock {\em Formal Methods in System Design}, 41(2):178--210, 2012.

\bibitem{alglave2009semantics}
Jade Alglave, Anthony Fox, Samin Ishtiaq, Magnus~O Myreen, Susmit Sarkar, Peter
  Sewell, and Francesco~Zappa Nardelli.
\newblock The semantics of power and arm multiprocessor machine code.
\newblock In {\em Proceedings of the 4th workshop on Declarative aspects of
  multicore programming}, pages 13--24. ACM, 2009.

\bibitem{alglave2013software}
Jade Alglave, Daniel Kroening, Vincent Nimal, and Michael Tautschnig.
\newblock Software verification for weak memory via program transformation.
\newblock In {\em Programming Languages and Systems}, pages 512--532. Springer,
  2013.

\bibitem{Alglave2011}
Jade Alglave and Luc Maranget.
\newblock {\em Computer Aided Verification: 23rd International Conference, CAV
  2011, Snowbird, UT, USA, July 14-20, 2011. Proceedings}, chapter Stability in
  Weak Memory Models, pages 50--66.
\newblock Springer Berlin Heidelberg, Berlin, Heidelberg, 2011.

\bibitem{alglave2014herding}
Jade Alglave, Luc Maranget, and Michael Tautschnig.
\newblock Herding cats: Modelling, simulation, testing, and data mining for
  weak memory.
\newblock {\em ACM Transactions on Programming Languages and Systems (TOPLAS)},
  36(2):7, 2014.

\bibitem{arm:llh}
ARM.
\newblock Cortex-{A9} {MPCore}\texttrademark, programmer advice notice,
  read-after-read hazards.
\newblock Technical report, 2011.
\newblock URL:
  \url{http://infocenter.arm.com/help/topic/com.arm.doc.uan0004a/UAN0004A\_a9\_read\_read.pdf}.

\bibitem{armv8ar}
ARM.
\newblock {\em ARM Architecture Reference Manual: ARMv8, for ARMv8-A
  architecture profile}.
\newblock 2017.

\bibitem{arvind2006memory}
Arvind and Jan-Willem Maessen.
\newblock Memory model = instruction reordering + store atomicity.
\newblock In {\em ACM SIGARCH Computer Architecture News}, volume~34, pages
  29--40. IEEE Computer Society, 2006.

\bibitem{Batty:2016:OSA:2914770.2837637}
Mark Batty, Alastair~F. Donaldson, and John Wickerson.
\newblock Overhauling sc atomics in c11 and opencl.
\newblock {\em SIGPLAN Not.}, 51(1):634--648, January 2016.

\bibitem{batty2011mathematizing}
Mark Batty, Scott Owens, Susmit Sarkar, Peter Sewell, and Tjark Weber.
\newblock Mathematizing c++ concurrency.
\newblock In {\em ACM SIGPLAN Notices}, volume~46, pages 55--66. ACM, 2011.

\bibitem{blundell2009invisifence}
Colin Blundell, Milo~MK Martin, and Thomas~F Wenisch.
\newblock Invisifence: performance-transparent memory ordering in conventional
  multiprocessors.
\newblock In {\em ACM SIGARCH Computer Architecture News}, volume~37, pages
  233--244. ACM, 2009.

\bibitem{boehm2008foundations}
Hans-J Boehm and Sarita~V Adve.
\newblock Foundations of the c++ concurrency memory model.
\newblock In {\em ACM SIGPLAN Notices}, volume~43, pages 68--78. ACM, 2008.

\bibitem{Boehm:2014:OGA:2618128.2618134}
Hans-J. Boehm and Brian Demsky.
\newblock Outlawing ghosts: Avoiding out-of-thin-air results.
\newblock In {\em Proceedings of the Workshop on Memory Systems Performance and
  Correctness}, MSPC '14, pages 7:1--7:6, New York, NY, USA, 2014. ACM.

\bibitem{cantin2003complexity}
Jason~F Cantin, Mikko~H Lipasti, and James~E Smith.
\newblock The complexity of verifying memory coherence.
\newblock In {\em Proceedings of the fifteenth annual ACM symposium on Parallel
  algorithms and architectures}, pages 254--255. ACM, 2003.

\bibitem{cenciarelli2007java}
Pietro Cenciarelli, Alexander Knapp, and Eleonora Sibilio.
\newblock The java memory model: Operationally, denotationally, axiomatically.
\newblock In {\em Programming Languages and Systems}, pages 331--346. Springer,
  2007.

\bibitem{ceze2007bulksc}
Luis Ceze, James Tuck, Pablo Montesinos, and Josep Torrellas.
\newblock Bulksc: bulk enforcement of sequential consistency.
\newblock In {\em ACM SIGARCH Computer Architecture News}, volume~35, pages
  278--289. ACM, 2007.

\bibitem{ReadersAndWriters1965}
Edsger~W. Dijkstra.
\newblock Cooperating sequential processes, technical report ewd-123.
\newblock Technical report, 1965.

\bibitem{dubois1986memory}
Michel Dubois, Christoph Scheurich, and Fay{\'e} Briggs.
\newblock Memory access buffering in multiprocessors.
\newblock In {\em ACM SIGARCH Computer Architecture News}, volume~14, pages
  434--442. IEEE Computer Society Press, 1986.

\bibitem{flur2016modelling}
Shaked Flur, Kathryn~E. Gray, Christopher Pulte, Susmit Sarkar, Ali Sezgin, Luc
  Maranget, Will Deacon, and Peter Sewell.
\newblock Modelling the armv8 architecture, operationally: Concurrency and isa.
\newblock In {\em Proceedings of the 43rd Annual ACM SIGPLAN-SIGACT Symposium
  on Principles of Programming Languages}, POPL 2016, pages 608--621, New York,
  NY, USA, 2016. ACM.

\bibitem{gharachorloo1991two}
Kourosh Gharachorloo, Anoop Gupta, and John~L Hennessy.
\newblock Two techniques to enhance the performance of memory consistency
  models.
\newblock In {\em Proceedings of the 1991 International Conference on Parallel
  Processing}, pages 355--364, 1991.

\bibitem{gharachorloo1990memory}
Kourosh Gharachorloo, Daniel Lenoski, James Laudon, Phillip Gibbons, Anoop
  Gupta, and John Hennessy.
\newblock Memory consistency and event ordering in scalable shared-memory
  multiprocessors.
\newblock In {\em Proceedings of the 17th International Symposium on Computer
  Architecture}, pages 15--26. ACM, 1990.

\bibitem{gniady2002speculative}
Chris Gniady and Babak Falsafi.
\newblock Speculative sequential consistency with little custom storage.
\newblock In {\em Parallel Architectures and Compilation Techniques, 2002.
  Proceedings. 2002 International Conference on}, pages 179--188. IEEE, 2002.

\bibitem{goodman1991cache}
James~R Goodman.
\newblock {\em Cache consistency and sequential consistency}.
\newblock University of Wisconsin-Madison, Computer Sciences Department, 1991.

\bibitem{gope2014atomic}
Dibakar Gope and Mikko~H Lipasti.
\newblock Atomic sc for simple in-order processors.
\newblock In {\em High Performance Computer Architecture (HPCA), 2014 IEEE 20th
  International Symposium on}, pages 404--415. IEEE, 2014.

\bibitem{guiady1999sc+}
Chris Guiady, Babak Falsafi, and Terani~N Vijaykumar.
\newblock Is sc+ ilp= rc?
\newblock In {\em Computer Architecture, 1999. Proceedings of the 26th
  International Symposium on}, pages 162--171. IEEE, 1999.

\bibitem{power2013version}
IBM.
\newblock {\em Power ISA, Version 2.07}.
\newblock 2013.

\bibitem{c11}
{International Organization for Standardization (ISO)}.
\newblock Information technology -- programming languages -- {C}, {ISO}/{IEC}
  9899:2011.
\newblock Technical report, December 2011.

\bibitem{alloy}
Daniel Jackson.
\newblock Alloy: A lightweight object modelling notation.
\newblock In {\em ACM Transactions on Software Engineering and Methodology
  (TOSEM)}, volume~11, April 2002.
\newblock URL: \url{http://alloy.mit.edu}.

\bibitem{Kang:2015:FCM:2737924.2738005}
Jeehoon Kang, Chung-Kil Hur, William Mansky, Dmitri Garbuzov, Steve Zdancewic,
  and Viktor Vafeiadis.
\newblock A formal c memory model supporting integer-pointer casts.
\newblock In {\em Proceedings of the 36th ACM SIGPLAN Conference on Programming
  Language Design and Implementation}, PLDI '15, pages 326--335, New York, NY,
  USA, 2015. ACM.

\bibitem{lahav2017repairing}
Ori Lahav, Viktor Vafeiadis, Jeehoon Kang, Chung-Kil Hur, and Derek Dreyer.
\newblock Repairing sequential consistency in {C}/{C++}11.
\newblock {\em 38th ACM SIGPLAN Conference on Programming Language Design and
  Implementation (PLDI)}, 2017.

\bibitem{lamport1979make}
Leslie Lamport.
\newblock How to make a multiprocessor computer that correctly executes
  multiprocess programs.
\newblock {\em Computers, IEEE Transactions on}, 100(9):690--691, 1979.

\bibitem{lin2012efficient}
Changhui Lin, Vijay Nagarajan, Rajiv Gupta, and Bharghava Rajaram.
\newblock Efficient sequential consistency via conflict ordering.
\newblock In {\em ACM SIGARCH Computer Architecture News}, volume~40, pages
  273--286. ACM, 2012.

\bibitem{lustig2017automated}
Daniel Lustig, Andrew Wright, Alexandros Papakonstantinou, and Olivier Giroux.
\newblock Automated generation of comprehensive memory model litmus test
  suites.
\newblock {\em 22nd ACM International Conference on Architectural Support for
  Programming Languages and Operating Systems (ASPLOS)}, 2017.

\bibitem{mador2012axiomatic}
Sela Mador-Haim, Luc Maranget, Susmit Sarkar, Kayvan Memarian, Jade Alglave,
  Scott Owens, Rajeev Alur, Milo~MK Martin, Peter Sewell, and Derek Williams.
\newblock An axiomatic memory model for power multiprocessors.
\newblock In {\em Computer Aided Verification}, pages 495--512. Springer, 2012.

\bibitem{maessen2000improving}
Jan-Willem Maessen, Arvind, and Xiaowei Shen.
\newblock Improving the java memory model using crf.
\newblock {\em ACM SIGPLAN Notices}, 35(10):1--12, 2000.

\bibitem{manson2005java}
Jeremy Manson, William Pugh, and Sarita~V. Adve.
\newblock The java memory model.
\newblock In {\em Proceedings of the 32Nd ACM SIGPLAN-SIGACT Symposium on
  Principles of Programming Languages}, POPL '05, pages 378--391, New York, NY,
  USA, 2005. ACM.

\bibitem{maranget2012tutorial}
Luc Maranget, Susmit Sarkar, and Peter Sewell.
\newblock A tutorial introduction to the arm and power relaxed memory models.
\newblock \url{http://www.cl.cam.ac.uk/~pes20/ppc-supplemental/test7.pdf},
  2012.

\bibitem{Nienhuis:2016:OSC:2983990.2983997}
Kyndylan Nienhuis, Kayvan Memarian, and Peter Sewell.
\newblock An operational semantics for c/c++11 concurrency.
\newblock In {\em Proceedings of the 2016 ACM SIGPLAN International Conference
  on Object-Oriented Programming, Systems, Languages, and Applications}, OOPSLA
  2016, pages 111--128, New York, NY, USA, 2016. ACM.

\bibitem{operationalC++}
Kyndylan Nienhuis, Kayvan Memarian, and Peter Sewell.
\newblock An operational semantics for c/c++11 concurrency.
\newblock In {\em Proceedings of the 2016 ACM SIGPLAN International Conference
  on Object-Oriented Programming, Systems, Languages, and Applications}, OOPSLA
  2016, pages 111--128, New York, NY, USA, 2016. ACM.

\bibitem{owens2009better}
Scott Owens, Susmit Sarkar, and Peter Sewell.
\newblock A better x86 memory model: x86-tso.
\newblock In {\em Theorem Proving in Higher Order Logics}, pages 391--407.
  Springer, 2009.

\bibitem{ranganathan1997using}
Parthasarathy Ranganathan, Vijay~S Pai, and Sarita~V Adve.
\newblock Using speculative retirement and larger instruction windows to narrow
  the performance gap between memory consistency models.
\newblock In {\em Proceedings of the ninth annual ACM symposium on Parallel
  algorithms and architectures}, pages 199--210. ACM, 1997.

\bibitem{sarkar2012synchronising}
Susmit Sarkar, Kayvan Memarian, Scott Owens, Mark Batty, Peter Sewell, Luc
  Maranget, Jade Alglave, and Derek Williams.
\newblock Synchronising c/c++ and power.
\newblock In {\em ACM SIGPLAN Notices}, volume~47, pages 311--322. ACM, 2012.

\bibitem{sarkar2011understanding}
Susmit Sarkar, Peter Sewell, Jade Alglave, Luc Maranget, and Derek Williams.
\newblock Understanding power multiprocessors.
\newblock In {\em ACM SIGPLAN Notices}, volume~46, pages 175--186. ACM, 2011.

\bibitem{Sarkar:2009:SXM:1594834.1480929}
Susmit Sarkar, Peter Sewell, Francesco~Zappa Nardelli, Scott Owens, Tom Ridge,
  Thomas Braibant, Magnus~O. Myreen, and Jade Alglave.
\newblock The semantics of x86-cc multiprocessor machine code.
\newblock {\em SIGPLAN Not.}, 44(1):379--391, January 2009.

\bibitem{sewell2010x86}
Peter Sewell, Susmit Sarkar, Scott Owens, Francesco~Zappa Nardelli, and
  Magnus~O Myreen.
\newblock x86-tso: a rigorous and usable programmer's model for x86
  multiprocessors.
\newblock {\em Communications of the ACM}, 53(7):89--97, 2010.

\bibitem{shen1999commit}
Xiaowei Shen, Arvind, and Larry Rudolph.
\newblock Commit-reconcile and fences (crf): A new memory model for architects
  and compiler writers.
\newblock In {\em Computer Architecture, 1999. Proceedings of the 26th
  International Symposium on}, pages 150--161. IEEE, 1999.

\bibitem{singh2012end}
Abhayendra Singh, Satish Narayanasamy, Daniel Marino, Todd Millstein, and
  Madanlal Musuvathi.
\newblock End-to-end sequential consistency.
\newblock In {\em ACM SIGARCH Computer Architecture News}, volume~40, pages
  524--535. IEEE Computer Society, 2012.

\bibitem{c++n4527}
Richard Smith, editor.
\newblock {\em Working Draft, Standard for Programming Language C++}.
\newblock \url{http://open-std.org/JTC1/SC22/WG21/docs/papers/2015/n4527.pdf},
  May 2015.

\bibitem{sparc1992sparcv8}
{SPARC International, Inc.}
\newblock {\em The SPARC Architecture Manual: Version 8}.
\newblock Prentice-Hall, Inc., 1992.

\bibitem{weaver1994sparc}
David~L Weaver and Tom Gremond.
\newblock {\em The SPARC architecture manual (Version 9)}.
\newblock PTR Prentice Hall Englewood Cliffs, NJ 07632, 1994.

\bibitem{wenisch2007mechanisms}
Thomas~F Wenisch, Anastasia Ailamaki, Babak Falsafi, and Andreas Moshovos.
\newblock Mechanisms for store-wait-free multiprocessors.
\newblock In {\em ACM SIGARCH Computer Architecture News}, volume~35, pages
  266--277. ACM, 2007.

\bibitem{memalloy}
John Wickerson, Mark Batty, Tyler Sorensen, and George~A. Constantinides.
\newblock Automatically comparing memory consistency models.
\newblock {\em 44th ACM SIGPLAN Symposium on Principles of Programming
  Languages (POPL)}, 2017.

\bibitem{wmm}
Sizhuo Zhang, Muralidaran Vijayaraghavan, and Arvind.
\newblock Weak memory models: Balancing definitional simplicity and
  implementation flexibility.
\newblock In {\em Proceedings of the 2017 International Conference on Parallel
  Architectures and Compilation}, Portland, OR, USA, 2017.

\end{thebibliography}

\appendix

\newpage
\section{GAM Operational model $\subseteq$ GAM Axiomatic Model}\label{sec:gam_axi_contain_op}

\begin{lemma}\label{lem:rob:ppo}
For any operational model state, the following properties hold:
\begin{enumerate}
    \item If $I_1\RobProgOrd I_2$, then whether $I_1$ is ordered before $I_2$ in any of $\RobDataDep,\RobAddrDep,\RobNTPreservePO$ only depends on the states of $I_1$, $I_2$, and instructions between $I_1$ and $I_2$ in the ROB.
    \item If we add a new instruction to the end of an ROB, then changes in $\RobDataDep,\RobAddrDep,\RobNTPreservePO$ can only involve new edges pointing to the newly added instruction.
    \item If we mark a not-done non-branch instruction as done in an ROB, then there is no change in $\RobDataDep,\RobAddrDep,\RobNTPreservePO$.
    \item If we mark a not-done branch instruction as done in an ROB, then the changes in $\RobDataDep,\RobAddrDep,\RobNTPreservePO$ can only involve removing existing edges.
    \item If we compute the store data of a store in an ROB, then there is no change in $\RobDataDep,\RobAddrDep,\RobNTPreservePO$.
    \item If we compute the address of a memory instruction in an ROB, then the changes in $\RobDataDep,\RobAddrDep$ can only involve removing existing edges.
    \item If we compute the address of a load $L$ to be $a$ in an ROB, then the changes in $\RobNTPreservePO$ can only involve:
    \begin{enumerate}
        \item new edges pointing to $L$, and
        \item new edges from $L$, and
        \item removal of existing edges.
    \end{enumerate}
    \item If we compute the address of a store $S$ to be $a$ in an ROB, then the changes in $\RobNTPreservePO$ can only involve:
    \begin{enumerate}
        \item new edges point to $S$, and
        \item new edges starting from $S$, and
        \item new edges across $S$, and
        \item removal of existing edges.
    \end{enumerate}
\end{enumerate}
\end{lemma}
\begin{proof}
    The cases in the lemma can be proved easily one by one.
\end{proof}

\begin{lemma}
    The following invariants hold during the execution of the operational model:
    \begin{enumerate}
        \item \label{inv:rob:ppo} If $I_1 \RobNTPreservePO I_2$ and $I_2.\DoneTS \neq \top$, then $I_1.\DoneTS \neq \top$ and $I_1.\DoneTS < I_2.\DoneTS$.
        \item \label{inv:rob:addr} If $I_1 \RobAddrDep I_2$ and $I_2.\AddrTS \neq \top$, then $I_1.\DoneTS \neq \top$ and $I_1.\DoneTS < I_2.\AddrTS$.
        \item \label{inv:rob:st-data} If $I_1 \RobDataDep I_2$, and  not $I_1\RobAddrDep I_2$, and $I_2$ is a store, and $I_2.\StDataTS \neq \top$, then $I_1.\DoneTS \neq \top$ and $I_1.\DoneTS < I_2.\StDataTS$.
        \item \label{inv:rob:addr-st} If $I_1 \RobProgOrd I_2$, and $I_1$ is a memory instruction, and $I_2$ is a store, and $I_2.\DoneTS \neq \top$, then $I_1.\AddrTS \neq \top$ and $I_1.\AddrTS < I_2.\DoneTS$.
        \item \label{inv:rob:kill-done-st} We never kill a done store.
        \item \label{inv:rob:mem-val} For any address $a$, let $S$ be the store with the maximum $\DoneTS$ among all the done stores for address $a$.
        The monolithic memory value for $a$ is equal to $S.\StData$.
        \item \label{inv:rob:rf} For any done load $L$, let $S = L.\FromSt$ (i.e., $S$ is the store read by $L$).
        All of the following properties are satisfied:
        \begin{enumerate}
            \item \label{prop:rob:rf:no-kill} $S$ still exists in an ROB (i.e., S is not killed).
            \item \label{prop:rob:rf:addr-data} $S.\MemAddr = L.\MemAddr$ and $S.\StData = L.\LdVal$.
            \item \label{prop:rob:rf:done-no-st} If $S$ is done, then there is no not-done store $S'$ such that $S'.addr = a$ and $S'\RobProgOrd L$.
            \item \label{prop:rob:rf:done-max} If $S$ is done, then for any other done store $S'$ with $S'.\MemAddr = L.\MemAddr$, if $S'\RobProgOrd L$ or $S'.\DoneTS < L.\DoneTS$, then $S'.\DoneTS < S.\DoneTS$.
            \item \label{prop:rob:rf:not-done} If $S$ is not done, then $S\RobProgOrd L$, and there is no store $S'$ such that $S'.\MemAddr = L.\MemAddr$ and $S\RobProgOrd S'\RobProgOrd L$.
        \end{enumerate}
    \end{enumerate}
\end{lemma}
\begin{proof}
We now prove the invariants inductively.
That is, when rule $R$ fires in the operational model, we assume that all the invariants hold before $R$ fires, and try to prove that the invariants still hold after $R$ fires.
To avoid confusion, we add superscript 0 to model states and orderings before $R$ fires, and add superscript 1 to model states and orderings after $R$ fires.
For example, $\RobNTPreservePO^0$ denotes the non-transitive preserved program order before $R$ fires, while $I.\DoneTS^1$ denotes the done timestamp of instruction $I$ after $R$ fires.
Consider the type of rule $R$:
\begin{enumerate}
    \item Fetch: Assume $R$ fetches a new instruction $I$ into an ROB.
    According to Lemma~\ref{lem:rob:ppo}, new edges in $\RobDataDep^1, \RobAddrDep^1, \RobNTPreservePO^1$ must point to $I$.
    Now we consider each invariant:
    \begin{itemize}
        \item Invariants~\ref{inv:rob:ppo}, \ref{inv:rob:addr} and \ref{inv:rob:st-data}:
        These invariants may be affected by the new edges in $\RobDataDep^1, \RobAddrDep^1, \RobNTPreservePO^1$.
        However, since $I.\DoneTS^1$, $I.\AddrTS^1$ and $I.\StDataTS^1$ are all $\top$, these invariants cannot be affected.
        
        \item Invariants~\ref{inv:rob:addr-st}, \ref{inv:rob:kill-done-st}, \ref{inv:rob:mem-val}, \ref{inv:rob:rf}: These invariants cannot be affected.
    \end{itemize}

    \item Execute-Reg-to-Reg: Assume $R$ executes a reg-to-reg instruction $I$ and marks it as done.
    According to Lemma~\ref{lem:rob:ppo}, $\RobDataDep^1, \RobAddrDep^1, \RobNTPreservePO^1$ are the same as $\RobDataDep^0, \RobAddrDep^0, \RobNTPreservePO^0$, respectively.
    $I.\DoneTS$ changes from $\top$ to current global time.
    Now we consider each invariant:
    \begin{itemize}
        \item Invariant~\ref{inv:rob:ppo}: This invariant can be affected by the change in $I.\DoneTS$.
        Consider any $I_1$ such that $I_1 \RobNTPreservePO^1 I$.
        Since $I$ is a reg-to-reg instruction, it must be that $I_1\RobDataDep^1 I$ according to the definition of $\RobNTPreservePO$.
        Since $\RobDataDep^0 = \RobDataDep^1$, $I_1\RobDataDep^0 I$.
        The guard of $R$ requires that $I_1$ is already done before $R$ fires, so $I_1.\DoneTS^1 = I_1.\DoneTS^0 < I.\DoneTS^1$, i.e., the invariant still holds.
        
        \item Invariants~\ref{inv:rob:addr}, \ref{inv:rob:st-data}, \ref{inv:rob:addr-st}, \ref{inv:rob:kill-done-st}, \ref{inv:rob:mem-val}, \ref{inv:rob:rf}: These invariants cannot be affected.
    \end{itemize}

    \item Execute-Branch: Assume $R$ executes a branch instruction $I$ and marks it as done.
    According to Lemma~\ref{lem:rob:ppo}, $\RobDataDep^1, \RobAddrDep^1, \RobNTPreservePO^1$ are contained by $\RobDataDep^0, \RobAddrDep^0, \RobNTPreservePO^0$, respectively. 
    $I.\DoneTS$ changes from $\top$ to current global time, and instructions younger than $I$ in the ROB may all be killed.
    Now we consider each invariant:
    \begin{enumerate}
        \item Invariant~\ref{inv:rob:ppo}:
        Consider any $I_1$ such that $I_1\RobNTPreservePO^1 I$.
        Since $I$ is a branch, it must be that $I_1\RobDataDep^1 I$.
        Since $\RobDataDep^1 \subseteq \RobDataDep^0$, $I_1\RobDataDep^0 I$.
        The guard of $R$ ensures that $I_1$ must be already done right before $R$ fires, so $I_1.\DoneTS^1 = I_1.\DoneTS^0 < I.\DoneTS^1$, i.e., the invariant still holds.
        
        \item Invariants~\ref{inv:rob:addr}, \ref{inv:rob:st-data}, \ref{inv:rob:addr-st}: These invariants cannot be affected.
        
        \item Invariant~\ref{inv:rob:kill-done-st}:
        This invariant may be affected if instructions are killed.
        We prove by contradiction, i.e., we assume a done store $S$ is killed in rule $R$.
        Since $S$ is killed, $I\RobProgOrd^0 S\Rightarrow I\RobNTPreservePO^0 S$.
        According to invariant~\ref{inv:rob:ppo}, $S.\DoneTS^0 \neq \top \Rightarrow I.\DoneTS^0 \neq \top$, i.e., $I$ is done even before $R$ fires.
        This contradicts with the guard of $R$.
        
        \item Invariant~\ref{inv:rob:mem-val}: This invariant cannot be affected.
        
        \item Invariant~\ref{inv:rob:rf}: 
        We consider each case in this invariant:
        \begin{itemize}
            \item Invariant~\ref{prop:rob:rf:no-kill}: This invariant can be affected by instruction kills.
            Assume a store $S$ is killed by $I$ in rule $R$ ($I\RobProgOrd^0 S$), and $S$ is read by a load $L$ (i.e., $S = L.\FromSt$).
            We have shown that $S$ cannot be done, so $S$ is not done before $R$ fires.
            Invariant~\ref{prop:rob:rf:not-done} says that $S\RobProgOrd^0 L$.
            Then $L$ will also be killed by $I$, so this invariant still holds.
            
            \item Invariants~\ref{prop:rob:rf:addr-data}, \ref{prop:rob:rf:done-no-st}, \ref{prop:rob:rf:done-max}, \ref{prop:rob:rf:not-done}: These invariants cannot be affected.
        \end{itemize}
    \end{enumerate}

    \item Execute-Fence: Assume $R$ executes a fence instruction $F$ and marks it as done.
    According to Lemma~\ref{lem:rob:ppo}, $\RobDataDep^1, \RobAddrDep^1, \RobNTPreservePO^1$ are the same as $\RobDataDep^0, \RobAddrDep^0, \RobNTPreservePO^0$, respectively.
    $I.\DoneTS$ changes from $\top$ to current global time.
    Now we consider each invariant:
    \begin{itemize}
        \item Invariant~\ref{inv:rob:ppo}: This invariant can be affected by the change in $I.\DoneTS$.
        Consider any $I_1$ such that $I_1 \RobNTPreservePO^1 F$.
        Since $\RobNTPreservePO^0 = \RobNTPreservePO^1$, $I_1 \RobNTPreservePO^0 F$.
        Since $F$ is a fence, it must be that $\OrderedFenceFunc(I_1, F)$ is true before $R$ fires.
        Then the guard of $R$ ensures that $I_1$ must be already done before $R$ fires, so $I_1.\DoneTS^1 = I_1.\DoneTS^0 < F.\DoneTS^1$, i.e., the invariant still holds.
        
        \item Invariants~\ref{inv:rob:addr}, \ref{inv:rob:st-data}, \ref{inv:rob:addr-st}, \ref{inv:rob:kill-done-st}, \ref{inv:rob:mem-val}, \ref{inv:rob:rf}: These invariants cannot be affected.
    \end{itemize}

    \item Compute-Store-Data: Assume $R$ computes the data of a store $S$.
    According to Lemma~\ref{lem:rob:ppo}, $\RobDataDep^1, \RobAddrDep^1, \RobNTPreservePO^1$ are the same as $\RobDataDep^0, \RobAddrDep^0, \RobNTPreservePO^0$, respectively.
    $S.\StDataTS$ changes from $\top$ to current global time.
    Now we consider each invariant:
    \begin{itemize}
        \item Invariants~\ref{inv:rob:ppo}, \ref{inv:rob:addr}: These invariant cannot be affected.
        
        \item Invariant~\ref{inv:rob:st-data}:
        Consider any instruction $I_1$ such that $I_1 \RobDataDep^1 S$ but not $I_1\RobAddrDep^1 S$.
        Note that $\RobDataDep^0 = \RobDataDep^1$ and $\RobAddrDep^0 = \RobAddrDep^1$.
        Therefore, the computation of the data of $S$ uses the result of $I_1$ as a source operand.
        Then the guard of $R$ ensures that $I_1$ must be already done before $R$ fires.
        Therefore, we have $I_1.\DoneTS^1 = I_1.\DoneTS^0 < S.\StDataTS^1$, and the invariant still holds.
        
        \item Invariants~\ref{inv:rob:addr-st}, \ref{inv:rob:kill-done-st}, \ref{inv:rob:mem-val}: These invariant cannot be affected.
        
        \item Invariant~\ref{inv:rob:rf}:
        If there exists a done load $L$ such that $L.\FromSt^0 = S$, then $S.\StData^0 \neq \top$ according to invariant~\ref{prop:rob:rf:addr-data}.
        That is, the store data of $S$ is already computed before $R$ fires, contradicting with the guard of $R$.
        Therefore, $S$ is not read by any load yet, and this invariant is not affected.
    \end{itemize}

    \item Execute-Store: Assume $R$ executes a store $S$ and marks it as done.
    According to Lemma~\ref{lem:rob:ppo}, $\RobDataDep^1, \RobAddrDep^1, \RobNTPreservePO^1$ are the same as $\RobDataDep^0, \RobAddrDep^0, \RobNTPreservePO^0$, respectively.
    $S.\DoneTS$ changes from $\top$ to current global time, and the monolithic memory is also updated.
    Now we consider each invariant:
    \begin{itemize}
        \item Invariant~\ref{inv:rob:ppo}: This invariant can be affected by the change in $S.\DoneTS$.
        Consider any $I_1$ such that $I_1\RobNTPreservePO^1 S$.
        Since $\RobNTPreservePO^0 = \RobNTPreservePO^1$, $I_1\RobNTPreservePO^0 S$.
        Since $S$ is a store, there can following cases to form $I_1\RobNTPreservePO^0 S$:
        \begin{itemize}
            \item $I_1\DataDep^0 S$: 
            The guard of $R$ ensures that $\max(S.\AddrTS^0, S.\StDataTS^0) < S.\DoneTS^1$.
            Invariants~\ref{inv:rob:addr} and \ref{inv:rob:st-data} says that $I_1.\DoneTS^0 < \max(S.\AddrTS^0, S.\StDataTS^0)$.
            Thus, we have $I_1.\DoneTS^1 = I_1.\DoneTS^0 < S.\DoneTS^1$, i.e., the invariant still holds.
            
            \item $I_1$ is a branch: The guard of $R$ ensures that $I$ is already done before $R$ fires, so $I_1.\DoneTS^1 = I_1.\DoneTS^0 < S.\DoneTS^1$, i.e., the invariant still holds.
            
            \item There exists a memory instruction $I$ such that $I_1 \RobAddrDep^0 I\RobProgOrd^0 S$:
            The guard of $R$ ensures that the address of $I$ has been computed before $R$ fires, i.e., $I.\AddrTS^0 < S.\DoneTS^1$.
            Invariant~\ref{inv:rob:addr} says that $I_1.\DoneTS^0 < I.\AddrTS^0$.
            Thus, $I_1.\DoneTS^1 = I_1.\DoneTS^0 < S.\DoneTS^1$, i.e., the invariant still holds.
            
            \item $I_1$ is a load whose address has been computed to the same as the address of $S$:
            The guard of $R$ ensures that $I$ is done before $R$ fires,  so $I_1.\DoneTS^1 = I_1.\DoneTS^0 < S.\DoneTS^1$, i.e., the invariant still holds.
            
            \item $I_1$ is a store whose address has been computed to the same as the address of $S$:
            The guard of $R$ ensures that $I$ is done before $R$ fires,  so $I_1.\DoneTS^1 = I_1.\DoneTS^0 < S.\DoneTS^1$, i.e., the invariant still holds.
            
            \item $\OrderedFunc(I_1, S)$ is true:
            The guard of $R$ ensures that $I$ is done before $R$ fires, so $I_1.\DoneTS^1 = I_1.\DoneTS^0 < S.\DoneTS^1$, i.e., the invariant still holds.  
        \end{itemize}
        
        \item Invariants~\ref{inv:rob:addr}, \ref{inv:rob:st-data}: These invariants are not affected.
        
        \item Invariants~\ref{inv:rob:addr-st}: This invariant can be affected by the change in $S.\DoneTS$.
        Consider any memory instruction $I_1$ such that $I_1\RobProgOrd^1 I_2$.
        Note that $\RobProgOrd$ cannot be changed by $R$, so $I_1\RobProgOrd^0 I_2$.
        The guard of $R$ ensures that $I_1$ has computed its address before $R$ fires, so $I_1.\AddrTS^1 = I_1.\AddrTS^0 < S.\DoneTS^1$, i.e., the invariant still holds.
        
        \item Invariant~\ref{inv:rob:kill-done-st}: This invariant is not affected.
        
        \item Invariant~\ref{inv:rob:mem-val}: This invariant can be affected by making $S$ done and updating the monolithic memory.
        We only need to focus on memory address $a = S.\MemAddr^0$; other addresses are not affected.
        Note that $S.\DoneTS^1$ is the maximum among the non-$\top$ $\DoneTS$ of every instruction.
        Therefore, after $R$ fires, $S$ is the store with the maximum $\DoneTS$ among all done stores for $a$.
        On the other hand, the monolithic memory location $a$ is updated to $S.\StData$ by rule $R$.
        Thus, the invariant still holds.
        
        \item Invariant~\ref{inv:rob:rf}:
        We assume $a = S.\MemAddr$.
        We consider each case in this invariant:
        \begin{itemize}
            \item Invariant~\ref{prop:rob:rf:no-kill}, \ref{prop:rob:rf:addr-data}: These invariants are not affected.
            
            \item Invariant~\ref{prop:rob:rf:done-no-st}: This invariant can be affected when there exists a done load $L$ such that $L.\FromSt^0 = S$.
            We need to show that there is no not-done store $S'$ such that $S'.\MemAddr^1 = a$ and $S'\RobProgOrd^1 L$.
            Since rule does not compute any address or change $\RobProgOrd$.
            It is equivalent to show that there is no not-done store $S'$ such that $S'.\MemAddr^0 = a$ and $S'\RobProgOrd^0 L$.
            We prove by contradiction, i.e., we assume such $S'$ exists before $R$ fires.
            Since $S$ is not done before $R$ fires, according to invariant~\ref{prop:rob:rf:not-done}, $S \RobProgOrd^0 L$ and it cannot be that $S\RobProgOrd^0 S' \RobProgOrd^0 L$.
            Thus, it must be that $S'\RobProgOrd S$.
            However, the guard of $R$ requires that $S'$ to be done bofore $R$ fires, contradicting with the assumption that $S'$ is not done.
            Thus, the invariant holds.
            
            \item Invariant~\ref{prop:rob:rf:done-max}: This invariant can be affected in the following two ways:
            \begin{enumerate}
                \item There exists a done load $L$ such that $L.\FromSt^0 = S$:
                In this case, we need to show that for any other done store $S'$ with $S'.\MemAddr^1 = a$, if $S'\RobProgOrd^1 L$ or $S'.\DoneTS^1 < L.\DoneTS^1$, then $S'.\DoneTS^1 < S.\DoneTS^1$.
                Since $S.\DoneTS^1$ is the maximum $\DoneTS$ among all done instructions, this invariant still holds.
                
                \item There exists a done load $L^*$ and a done store $S^*$ such that $L^*.\FromSt^1 = S^*$ and $L^*.\MemAddr^1 = S^*.\MemAddr^1 = S^*$:
                In this case, we need to show that if $S\RobProgOrd^1 L^*$ or $S.\DoneTS^1 < L^*.\DoneTS^1$, then $S.\DoneTS^1 < S^*.\DoneTS^1$.
                Since $S.\DoneTS^1$ is the maximum, $S.\DoneTS^1 < L^*.\DoneTS^1$ must be false.
                We will now show that $S\RobProgOrd^1 L^*$ is also impossible to prove the invariant holds.
                We prove it by contradiction, i.e., we assume $S\RobProgOrd^1 L^*$.
                Since rule $R$ does not change $\RobProgOrd$ or any store address or any state of $L^*$ and $S^*$, we have $S\RobProgOrd^0 L^*$, $S.\MemAddr^0 = a = L^*.\MemAddr^* = S^*.\MemAddr^0$, $L^*.\FromSt^0 = S^*$, and $S^*$ is done before $R$ fires.
                This contradicts with Invariant~\ref{prop:rob:rf:done-no-st}.
                Therefore, the invariant still holds.
            \end{enumerate}
            
            \item Invariant~\ref{prop:rob:rf:not-done}: This invariant is not affected.
        \end{itemize}
    \end{itemize}
    
    \item Execute-Load: Assume $R$ executes a load $L$.   
    If $L$ is not marked as done, the model state does not change, so all invariants still hold.
    Now we consider the case that $L$ is marked as done.
    According to Lemma~\ref{lem:rob:ppo}, $\RobDataDep^1, \RobAddrDep^1, \RobNTPreservePO^1$ are the same as $\RobDataDep^0, \RobAddrDep^0, \RobNTPreservePO^0$, respectively.
    $L.\DoneTS$ changes from $\top$ to current global time.
    Now we consider each invariant:
    \begin{itemize}
        \item Invariant~\ref{inv:rob:ppo}: This invariant can be affected by the change in $L.\DoneTS$.
        Assume $a = L.\MemAddr^0$.
        Consider any $I_1$ such that $I_1\RobNTPreservePO^1 L$.
        Since $\RobNTPreservePO$ is not changed by $R$, we have $I_1\RobNTPreservePO^0 L$.
        This ordering can be caused by the following cases:
        \begin{itemize}
            \item $I_1\DataDep^0 L$: Since $L$ only needs to compute the address from registers, $I_1\AddrDep^0 L$.
            Invariant~\ref{inv:rob:addr} says that $I_1.\DoneTS^0 < I_1.\AddrTS^0$.
            The guard of $R$ ensures that $L.\AddrTS^0 < L.\DoneTS^1$.
            Therefore $I_1.\DoneTS^1 = I_1.\DoneTS^0 < L.\DoneTS^1$, i.e., the invariant still holds.
                          
            \item There exists a store $S$ such that $S.\MemAddr^1 = a$ and $I_1 \RobDataDep^1 S\RobProgOrd^1 L$, and there is no store $S'$ such that $S'.\MemAddr^1 = a$ and $S \RobProgOrd^1 S' \RobProgOrd^1 L$:
            Since $R$ does not change $\RobProgOrd$, $\RobDataDep$ or any address, the above condition becomes: $S.\MemAddr^0 = a$, and $I_1\RobDataDep^0 S\RobProgOrd L$, and there is no store $S'$ such that $S'.\MemAddr^0 = a$ and $S \RobProgOrd^0 S' \RobProgOrd^0 L$.
            Before $R$ fires, if there are not-done loads with computed addresses $a$ between $S$ and $L$ in the ROB, then let $L'$ be the youngest of them in ROB, and the ROB search conducted in $R$ will stop at $L'$.
            This will make $R$ not mark $L$ as done, contradicting with our previous assumption.
            Therefore, right before $R$ fires, any load with computed address $a$ between $S$ and $L$ in the ROB must be done.
            Since $S.\MemAddr^0 = a$, the ROB search in $R$ will search through $S$.
            If $S.\StDataTS^0 = \top$, then the ROB search will stop at $S$ and $L$ cannot be marked as done, contradicting with our assumption.
            Therefore, $S.\StDataTS\neq \top \Rightarrow S.\StDataTS^0 < L.\DoneTS^1$.
            Since address of $S$ is computed before $R$ fires, $S.\AddrTS^0 < L.\DoneTS^1$.
            Invariants~\ref{inv:rob:addr} and \ref{inv:rob:st-data} say that $I_1.\DoneTS^0 < \max(S.\AddrTS^0, S.\StDataTS^0)$, so $I_1.\DoneTS^1 = I_1.\DoneTS^0 < L.\DoneTS^1$, i.e., the invariant still holds.
            
            \item $I_1$ is a load with $I_1.\MemAddr^1 = a$, and there is no store $S$ such that $S.\MemAddr^1 = a$ and $I_1 \RobProgOrd^1 S \RobProgOrd^1 L$:
            Since $R$ does not change $\RobProgOrd$ or any address, the above condition becomes: $I_1.\MemAddr^0 = a$, and there is no store $S$ such that $S.\MemAddr^0 = a$ and $I_1 \RobProgOrd^0 S \RobProgOrd^1 L$.
            Before $R$ fires, if there are not-done loads with computed addresses $a$ between $I_1$ and $L$ in the ROB, then let $L'$ be the youngest of them in ROB, and the ROB search conducted in $R$ will stop at $L'$.
            This will make $R$ not mark $L$ as done, contradicting with our previous assumption.
            Therefore, right before $R$ fires, any load with computed address $a$ between $I_1$ and $L$ in the ROB must be done.
            Since $I_1.\MemAddr = a$, the ROB search in $R$ will search through $S$.
            If $I_1$ is not done before $R$ fires, then the ROB search will stop at $L$ and $L$ cannot be marked as done in $R$, contradicting with our previous assumption.
            Therefore $I_1$ is done before $R$ fires, so $I_1.\DoneTS^1 = I_1.\DoneTS^0 < L.\DoneTS^1$, i.e., the invariant still holds.
            
            \item $\OrderedFunc(I_1,L)$ is true:
            The guard of $R$ ensures that $I$ is done before $R$ fires, so $I_1.\DoneTS^1 = I_1.\DoneTS^0 < S.\DoneTS^1$, i.e., the invariant still holds.
        \end{itemize}
        
        \item Invariants~\ref{inv:rob:addr}, \ref{inv:rob:addr-st}, \ref{inv:rob:st-data}, \ref{inv:rob:kill-done-st}, \ref{inv:rob:mem-val}: These invariants cannot be affected.
        
        \item Invariant~\ref{inv:rob:rf}: This invariant can be affected because $L$ becomes done.
        Let $S = L.\FromSt^1$, and let $a = L.\MemAddr^1$.
        We need to show that $L$ and $S$ satisfies all the sub-invariants.
        We consider each case of separately.
        \begin{itemize}
            \item Invariant~\ref{prop:rob:rf:no-kill}: This invariant cannot be affected.
            
            \item Invariant~\ref{prop:rob:rf:addr-data}: We need to show that $S.\MemAddr^1 = L.\MemAddr^1$ ad $S.\StData^1 = L.\LdVal^1$.
            Note that $S.\MemAddr^0 = S.\MemAddr^0 = L.\MemAddr^0 = L.\MemAddr^1$.
            Also note that $S.\StData^0 = S.\StData^1$ and $L.\LdVal$ is the value read in rule $R$.
            If $L$ bypasses from local store in ROB, then $S$ is the store being bypassed, and the invariant holds.
            Otherwise, $L$ reads from monolithic memory, and invariant~\ref{inv:rob:mem-val} ensures that invariant holds.
            
            \item Invariant~\ref{prop:rob:rf:done-no-st}: Since $S$ is done after $R$ fires, $S$ is also done before $R$ fires.
            Then $R$ reads the value from monolithic memory.
            The guard of $R$ ensures that there is no store $S'$ such that $S'.\MemAddr^0 = a$ and $S'\RobProgOrd^0 L$.
            Since $R$ does not change address or $\RobProgOrd$, the invariant still holds.
            
            \item Invariant~\ref{prop:rob:rf:done-max}: Using the same argument as above, $R$ reads the value from monolithic memory.
            Right before $R$ fires, invariant~\ref{inv:rob:mem-val} says that for any other done store $S'$ with $S'.\MemAddr^0 = a$, $S'.\DoneTS < S.\DoneTS$.
            Since $R$ does not change $\DoneTS$ of stores or address, the invariant still holds.
            
            \item Invariant~\ref{prop:rob:rf:not-done}: Since $S.\DoneTS^1 = \top$, $S.\DoneTS^0 = \top$.
            Then $R$ reads the value by bypassing from $S$ in the local ROB, i.e., the ROB search in $R$ stops at $S$.
            We now prove by contradiction, i.e., we assume there exists $S'$ with $S'.\MemAddr^1 = a$ and $S\RobProgOrd^1 S' \RobProgOrd^1 L$.
            Since $R$ does not change address or $\RobProgOrd$, $S'.\MemAddr^0 =a$ and $S\RobProgOrd^0 S' \RobProgOrd^0 L$.
            Since $S\RobNTPreservePO^0 S'$, invariant~\ref{inv:rob:ppo} says that $S'.\DoneTS^0 = \top$.
            Then the ROB search cannot stop at $S$, contradicting with our previous conclusion.
            Therefore the invariant still holds.
        \end{itemize}
    \end{itemize}
    
    \item Compute-Mem-Addr for load $L$: $R$ computes the address of load $L$ to $a$.
    According to Lemma~\ref{lem:rob:ppo}, edges in $\RobDataDep$ and $\RobAddrDep$ may reduce.
    Some edges in $\RobNTPreservePO$ may be removed, but there can also be new $\RobNTPreservePO$ edges.
    $L.\MemAddr$ changes from $\top$ to $a$.
    We consider each invariant separately.
    \begin{itemize}
        \item Invariant~\ref{inv:rob:ppo}:
        Since $L$ is not done, we only need to consider newly generated $\RobNTPreservePO$ edges that start from $L$.
        Consider any $I_2$ such that $L\RobNTPreservePO^1 I_2$ but not $L\RobNTPreservePO^0 I_2$.
        We need to show that $I_2.\DoneTS^1 = \top$.
        $I_2$ must be in the following cases:
        \begin{itemize}
            \item $I_2$ is a store, $L\RobProgOrd^1 I_2$, and $I_2.\MemAddr^1 = a$:
            We prove by contradiction, i.e., assume $I_2.\DoneTS^1 \neq \top$.
            This gives $I_2.\DoneTS^0 \neq \top$.
            Since $L\RobProgOrd^0 S$ and $L.\AddrTS^0 = \top$, invariant~\ref{inv:rob:addr} says that $I_2.\DoneTS^0 = \top$, contradicting with previous assumption.
            Therefore the invariant still holds.
            
            \item $I_2$ is a load, $L\RobProgOrd^1 I_2$, $I_2.\MemAddr^1 = a$, and there is no store $S$ such that $S.\MemAddr^1 = a$ and $L\RobProgOrd^1 S\RobProgOrd^1 I_2$:
            We prove by contradiction, i.e., assume $I_2.\DoneTS^1 \neq \top$.
            This implies that $I_2$ is also done before $R$ fires.
            Since the ROB search in $R$ does not kill $I_2$, there must be a not-done load $L'$ such that $L'.\MemAddr^0 = a$ and $L\RobProgOrd^0 L' \RobProgOrd^0 I_2$.
            Now we have $L'\RobNTPreservePO I_2$.
            Since $L'.\DoneTS = \top$, this contradicts with invariant~\ref{inv:rob:ppo}.
            Therefore the invariant still holds.
        \end{itemize}
        
        \item Invariant~\ref{inv:rob:addr}: The guard of $R$ ensures that this invariant still holds.
        
        \item Invariants~\ref{inv:rob:st-data}, \ref{inv:rob:addr-st}: These invariants cannot be affected.
        
        \item Invariant~\ref{inv:rob:kill-done-st}: We prove by contradiction, i.e., we assume a done store $S$ is killed.
        Note that $S.\DoneTS^0 = S.\DoneTS^1 \neq \top$ and $L\RobProgOrd^0 S$.
        Invariant~\ref{inv:rob:addr-st} says that $L.\AddrTS^0 \neq \top$, contradicting with the guard of $R$.
        Therefore the invariant still holds.
        
        \item Invariants~\ref{inv:rob:mem-val} and \ref{inv:rob:rf}: These invariants cannot be affected.
    \end{itemize}

    \item Compute-Mem-Addr for store $S$: $R$ computes the address of $S$ to be $a$.
    According to Lemma~\ref{lem:rob:ppo}, edges in $\RobDataDep$ and $\RobAddrDep$ may reduce.
    Some edges in $\RobNTPreservePO$ may be removed, but there can also be new $\RobNTPreservePO$ edges.
    $S.\MemAddr$ changes from $\top$ to $a$.
    We consider each invariant separately.
    \begin{itemize}
        \item Invariant~\ref{inv:rob:ppo}:
        Since $S$ is not done, we do not need to consider new $\RobNTPreservePO$ edges ending at $S$.
        We only need to consider new $\RobNTPreservePO$ edges starting from $S$ or across $S$.
        There are the following two cases:
        \begin{itemize}
            \item There is store $S'$ such that $S'.\MemAddr^1 = a$ and $S\RobProgOrd^1 S'$:
            We need to show that $S'.\DoneTS^1 = \top$.
            We prove by contradiction, i.e., we assume $S'.\DoneTS^1 \neq \top$.
            This implies $S'.\DoneTS^0 \neq \top$ and $S\RobProgOrd^0 S'$.
            According to invariant~\ref{inv:rob:addr-st}, $S.\MemAddr^0 \neq \top$, contradicting with the guard of $R$.
            
            \item There is instruction $I_1$ and load $L$ such that $L.\MemAddr^1 = a$ and $I_1\RobDataDep^1 S\RobProgOrd^1 L$, and there is no store $S'$ such that $S'.\MemAddr^1=a$ and $S\RobProgOrd^1 S'\RobProgOrd^1 L$:
            The above statement becomes: $L.\MemAddr^0 = a$, $I_1\RobDataDep^0 S\RobProgOrd^0 L$, and there is not store $S'$ such that $S'.\MemAddr^0 = a$ and $S\RobProgOrd^0 S'\RobProgOrd^0 L$.
            We can show that $L.\DoneTS^1 = \top$, so the invariant still holds.
            We prove by contradiction, i.e., $L.\DoneTS^1\neq \top$.
            This implies $L.\DoneTS^0\neq \top$.
            Since $S$ is not killed by the ROB search in $R$, there must be $L'$ such that $L'.\MemAddr^0 = a$, $L'.\DoneTS^0 = \top$ and $S\RobProgOrd^0 L'\RobProgOrd^0 L$.
            This gives $L'\RobNTPreservePO^0 L$.
            Since $L.\DoneTS^0\neq \top$, this contradicts invariant~\ref{inv:rob:ppo}.
        \end{itemize}
        
        \item Invariant~\ref{inv:rob:addr}: The guard of $R$ ensures that this invariant still holds.
        
        \item Invariants~\ref{inv:rob:st-data}, \ref{inv:rob:addr-st}: These invariant cannot be affected.
        
        \item Invariant~\ref{inv:rob:kill-done-st}: We prove by contradiction, i.e., we assume a done store $S'$ is killed.
        That is, $S'.\DoneTS^0 \neq \top$ and $S\RobProgOrd^0 S'$.
        Invariant~\ref{inv:rob:addr-st} says that $S.\MemAddr^0 \neq \top$, contradicting with the guard of $R$.
        
        \item Invariant~\ref{inv:rob:mem-val}: This invariant is not affected.
        
        \item Invariant~\ref{inv:rob:rf}: For any done load $L^*$ for address $a$, assume $S^* = L^*.\FromSt^1$.
        The address computation of $S$ may prevent $L^*$ and $S^*$ from satisfying the invariants here.
        Note that $L^*.\MemAddr^0 = a$ and $S^* = L^*.\FromSt^0$.
        We consider each case of this invariant:
        \begin{itemize}
            \item Invariant~\ref{prop:rob:rf:no-kill}:
            We prove by contradiction, i.e., $S^*$ is killed but $L^*$ is not.
            We have proved that $S^*$ cannot be done, so $S^*$ is not done before $R$ fires.
            Invariant~\ref{prop:rob:rf:not-done} says that $S^* \RobProgOrd^0 L^*$.
            Then $L^*$ is also killed, contradicting with our assumption.
            
            \item Invariant~\ref{prop:rob:rf:addr-data}: This invariant cannot be affected.
            
            \item Invariant~\ref{prop:rob:rf:done-no-st}:
            We need to show that if $S^*$ is done, then it is impossible that $S \RobProgOrd^1 L^*$.
            We prove by contradiction, i.e., assume $S^*.\DoneTS^0 \neq \top$ and $S\RobProgOrd^1 L^*$.
            This implies that $S\RobProgOrd^0 L^*$..
            Invariant~\ref{prop:rob:rf:done-no-st} says that there is no store $S'$ such that $S'.\MemAddr^0 = a$ and $S'\RobProgOrd^0 L^*$.
            Since $L^*$ is not killed in the ROB search of rule $R$, there must be load $L'$ such that $L'.\MemAddr^0 = a$ and $L'.\DoneTS^0 = \top$ and $L'\RobProgOrd^0 L^*$.
            This gives $L' \RobNTPreservePO^0 L^* \Rightarrow L'.\DoneTS^0 \neq \top$, contradicting with previous conclusion.
            Therefore the invariant still holds.
            
            \item Invariant~\ref{prop:rob:rf:done-max}: This invariant is not affected.
            
            \item Invariant~\ref{prop:rob:rf:not-done}: We need to show that if $S^*$ is not done, then it is impossible that $S^* \RobProgOrd^1 S\RobProgOrd^1 L^*$.
            We prove by contradiction, i.e., we assume $S^*.\DoneTS^0\neq \top$ and $S^* \RobProgOrd^1 S\RobProgOrd^1 L^*$.
            This implies $S^* \RobProgOrd^0 S\RobProgOrd^0 L^*$.
            Invariant~\ref{prop:rob:rf:not-done} says that there is no store $S'$ such that $S'.\MemAddr^0 =a$ and $S^*\RobProgOrd^0 S' \RobProgOrd^0 L^*$.
            Since $L^*$ is not killed by the ROB search in rule $R$, there must be load $L'$ such that  $L'.\MemAddr^0 = a$ and $L'.\DoneTS^0 = \top$ and $S\RobProgOrd^0 L'\RobProgOrd^0 L^*$.
            This gives $L' \RobNTPreservePO^0 L^* \Rightarrow L'.\DoneTS^0 \neq \top$, contradicting with previous conclusion.
            Therefore the invariant still holds.
        \end{itemize}
    \end{itemize}

\end{enumerate}
\end{proof}

\newpage
\section{GAM Axiomatic model $\subseteq$ GAM Operational Model}\label{sec:gam_op_contain_axi}

\begin{theorem}
    GAM axiomatic model $\subseteq$ GAM operational model.
\end{theorem}
\begin{proof}
    The goal is that for any legal axiomatic relations $\langle \ProgOrd, \MemOrd, \ReadFrom\rangle$ (which satisfy the GAM axioms), we can run the operational model to give the same program behavior.
    The strategy to run the operational model consists of two major phases.
    In the first phase, we only fire Fetch rules to fetch all instructions into all ROBs according to $\ProgOrd$.
    During the second phase, in each step we fire a rule that either marks an instruction as done or computes the address or data of a memory instruction.
    Which rule to fire in a step depends on the current state of the operational model and $\MemOrd$.
    Here we give the detailed algorithm that determines which rule to fire in each step:
    \begin{enumerate}
        \item If in the operational model there is a not-done reg-to-reg or branch instruction whose source registers are all ready, then we fire an Execute-Reg-to-Reg or Execute-Branch rule to execute that instruction.
        \item If the above case does not apply, and in the operational model there is a memory instruction, whose address is not computed but the source registers for the address computation are all ready, then we fire a Compute-Mem-Addr rule to compute the address of that instruction.
        \item If neither of the above cases applies, and in the operational model there is a store instruction, whose store data is not computed but the source registers for the data computation are all ready, then we fire a Compute-Store-Data rule to compute the store data of that instruction.
        \item If none of the above cases applies, and in the operational model there is a fence instruction and the guard of the Execute-Fence rule for this fence is ready, then we fire the Execute-Fence rule to execute that fence.
        \item \label{sim:mem} If none of the above cases applies, then we find the oldest instruction in $\MemOrd$, which is not-done in the operational model, and we fire an Execute-Load or Execute-Store rule to execute that instruction.
    \end{enumerate}
    Before giving the invariants, we give a definition related  to the ordering of stores for the same address.
    For each address $a$, all stores for $a$ are totally ordered by $\MemOrd$, and we refer to this total order of stores for $a$ as $<_{co}^a$.
    
    Now we show the invariants.
    After each step, we maintain the following invariants:
    \begin{enumerate}
        \item The order of instructions in each ROB in the operational model is the same as the $\ProgOrd$ of that processor in the axiomatic relations.
        \item \label{inv:result} The results of all the instructions that have been marked as done so far in the operational model are the same as those in the axiomatic relations.
        \item All the load/store addresses that have been computed so far in the operational model are the same as those in the axiomatic relations.
        \item All the store data that have been computed so far in the operational model are the same as those in the axiomatic relations.
        \item \label{inv:no-kill} No kill has ever happened in the operational model.
        \item For the rule fired in each step that we have performed so far, the guard of the rule is satisfied the at that step (i.e., the rule can fire).
        \item \label{inv:done} In each step that we have performed so far, if we fire a rule to execute an instruction (especially a load) in that step, the instruction must be marked as done by the rule.
        \item \label{inv:store} For each address $a$, the order of all the store updates on monolithic memory address $a$ that have happened so far in the operational model is a prefix of $<_{co}^a$.
    \end{enumerate}
    We now examine each option that we may choose in each step of phase 2, and verify that all invariants hold.
    \begin{enumerate}
        \item We execute a reg-to-reg or branch instruction: trivial.

        \item We compute the address of a load or store instruction $I$: we only need to verify invariant~\ref{inv:no-kill}, i.e., this address computation will not kill any done load.
        Assume the address of $I$ is $a$, and $I$ is in processor $i$.   
        We prove by contradiction, i.e., we assume the address computation of $I$ kills a done load $L$ in the ROB of processor $i$.
        We search $\ProgOrd$ of processor $i$ from $L$ towards the oldest instruction (excluding $L$).
        We stop the search when we find a load or store instruction for address $a$ (note that all addresses in $\ProgOrd$ are known).
        We refer to the instruction found as $I_a$.
        Since $I$ has address $a$, either $I=I_a$ or $I \ProgOrd I_a$.
        If $I=I_a$, then $I_a$ is not-done and its address is not computed before the kill.
        In case of $I \ProgOrd I_a$, $I_a$ cannnot be a done store (because the address of $I$ is just computed).
        In this case, $I_a$ must be not-done and its address must not  be computed, otherwise $L$ will not be killed.
        In either case, $I_a$ must be not-done and the address of $I_a$ is not computed before the kill.
        We also have the following observation:
        \begin{itemize}
            \item $L$ can only become done via option~\ref{sim:mem} in a previous step, so for any not-done memory instruction $I'$ in the operational model, we have $L \MemOrd I'$ in the axiomatic relations.
        \end{itemize}
        We consider the following two possibilities:
        \begin{enumerate}
            \item $I_a$ is a load: In this case, the above observation says that $L \MemOrd I_a$.
            However, since there is no other memory instruction for address $a$ between $I_a$ and $L$ in processor $i$, $I_a \SameAddrOrd L \Rightarrow I_a \MemOrd L$.
            Thus this case is impossible.
            
            \item $I_a$ is a store: In this case, we consider which store is read by $L$.
            \begin{enumerate}
                \item $L$ bypasses from a store $S$ in processor $i$.
                $S$ must be older than $I_a$ in ROB (because address of $I_a$ is not computed at the time when $L$ is executed), so $S \ProgOrd I_a \Rightarrow S \MemOrd I_a$.
                This contradicts with the Load-Value axiom.
                
                \item $L$ reads the value of a store $S$ from the monolithic memory.
                Since all the done stores for $a$ form the prefix of $<_{co}^a$ and $I_a$ is not-done, $S <_{co}^a I_a \Rightarrow S\MemOrd I_a$.
                This also contradicts with the Load-Value axiom.
            \end{enumerate}
        \end{enumerate}
        Therefore this address computation cannot cause any kill.

        \item We compute a store data: trivial.
        
        \item We execute a fence instruction: trivial.

        \item We execute a memory instruction $I$ from $\MemOrd$: 
        First note that for any memory instruction $I'$ such that $I' \MemOrd I$, $I'$ must be done in the operational model (because of the way we pick $I$).
        We now prove according to the type of $I$.
        \begin{itemize}        
            \item $I$ is a store for address $a$: we first show that all the guards are satisfied:
            \begin{enumerate}
                \item Address and data of $I$ have been computed:
                We prove by contradiction, i.e., the address or data of $I$ has not been computed.
                We backtrack the dependency chain on $I$.
                The only reason for not being able to compute the address or data of $I$ is that an older instruction $I_1$ is not-done, and that $I_1 \DataDep I$.
                $I_1$ can only be a not-done reg-to-reg instruction or a not-done load.
                If $I_1$ is a reg-to-reg instruction, it cannot be computed because of a not-done instruction $I_2 \DataDep I_1$.
                We trace this dependency chain until we encounter a load, i.e., $I_k \DataDep I_{k-1} \DataDep \ldots \DataDep I_1 \DataDep I$, where $I_k$ is a not-done load and $I_1\ldots I_{k-1}$ are all not-done reg-to-reg instructions.
                Now we have $I_k \PreservePO I \Rightarrow I_k \MemOrd I$, contradicting with the way we pick $I$ in option~\ref{sim:mem}.
                
                \item All older memory or fence instructions that are ordered before $I$ by $\OrderedFunc$ are done:
                We prove by contradiction, i.e., assume there is a not-done memory or fence instruction $I_1$ that is older than $I$ in ROB and satisfies $\OrderedFunc(I_1, I)$.
                This implies $I_1 \ProgOrd I \Rightarrow I_1\LSFOrd I$.
                If $I_1$ is a not-done fence, the guard to execute it in the operational model must be false according to our algorithm.
                Therefore, there must be a not-done memory or fence instruction $I_2$ that is older than $I_1$ in ROB and satisfies $\OrderedFunc(I_2, I_1)$.
                This implies $I_2\LSFOrd I$.
                We keep backtracking if $I_2$ is also a not-done fence.
                We stop backtracking until $I_k$ is a not-done memory instruction.
                That is, we have $I_k \LSFOrd I_{k-1} \LSFOrd\cdots\LSFOrd I_1 \LSFOrd I \Rightarrow I_k \MemOrd I$.
                Since $I_k$ is a not-done memory instruction, this contradicts with the way we pick $I$ in option~\ref{sim:mem}.
                Therefore, such $I_1$ does not exist.
                
                \item All older branches are done:
                We prove by contradiction, i.e., an older branch $B$ in the ROB of $I$ is not-done.
                We backtrack the dependency chain on $B$, and get $I_k \DataDep I_{k-1} \DataDep \ldots \DataDep I_1 \DataDep B$, where $I_k$ is a not-done load and $I_1\ldots I_{k-1}$ are all not-done reg-to-reg instructions.
                Since $B \ProgOrd I$, we have $I_k \PreservePO B \DepOrd I \Rightarrow I_k \PreservePO I \Rightarrow I_k \MemOrd I$, contradicting with the way we pick $I$ in option~\ref{sim:mem}.
                
                \item All older loads and stores have computed their addresses:
                We prove by contradiction, i.e., an older load or store $I'$ in the ROB of $I$ has not computed its address.
                We backtrack the dependency chain on the address of $I'$, and get $I_k \DataDep I_{k-1} \DataDep \ldots \DataDep I_1 \AddrDep I'$, where $I_k$ is a not-done load and $I_1\ldots I_{k-1}$ are all not-done reg-to-reg instructions.
                Since $I' \ProgOrd I$, we have $I_k \PreservePO I_1 \DepOrd I \Rightarrow I_k \PreservePO I \Rightarrow I_k \MemOrd I$, contradicting with the way we pick $I$ in option~\ref{sim:mem}.
                
                \item All older loads and stores for address $a$ are done:
                For any store $S$ for address $a$ that is older than $I$ in the ROB of $I$, we have $S \SameAddrOrd I \Rightarrow S \MemOrd I$.
                Therefore, $S$ must be done.
                For any load $L$ for address $a$ that is older than $I$ in the ROB of $i$, we have $L \SameAddrOrd I\Rightarrow L \MemOrd I$.
                Therefore, $L$ must be done.
            \end{enumerate}
            We now only need to verify invariant~\ref{inv:store}.
            This is trivial, because all stores for $a$ that is older than $I$ in $<_{mo}$ are done (i.e., have updated $m[a]$).
            
            \item $I$ is a load for address $a$:
            we first show that all the guards are satisfied:
            \begin{enumerate}
                \item Address of $I$ has been computed: same argument as store case.
                \item All older memory or fence instructions that are ordered before $I$ by $\OrderedFunc$ are done: same argument as store case.
            \end{enumerate}
            We now need to verify invariants~\ref{inv:result} and \ref{inv:done}.
            To do this, we consider the three possible outcomes of the ROB search in the Execute-Load rule that executes $I$:
            \begin{enumerate}
                \item The search finds a not-done load $L$: We prove that this case is impossible (for invariant~\ref{inv:done}) by contradiction.
                If there are intervening stores for $a$ between $L$ and $I$ in the ROB, none of those stores can be done, because the not-done load $L$ will make the guards of Execute-Store rules fail.
                Let $S$ be the youngest store among these stores.
                $S$ must not have computed its address, because otherwise the search will stop at $S$.
                Now we backtrack the dependency chain on the address of $S$, and get $I_k \DataDep I_{k-1} \DataDep \ldots \DataDep I_1 \DataDep S$, where $I_k$ is a not-done load and $I_1\ldots I_{k-1}$ are all not-done reg-to-reg instructions.
                Since there is no store for $a$ between $S$ and $I$, we have $I_k \PreservePO I_1 \DepOrd I\Rightarrow I_k \PreservePO I \Rightarrow I_k \MemOrd I$.
                Since $I_k$ is not done, this contradicts with the way we pick $I$ in option~\ref{sim:mem}.
                Therefore, there is no store for $a$ between $L$ and $I$ in the ROB.
                Then we have $L \SameAddrOrd I \Rightarrow L \MemOrd I$.
                Since $L$ is not done, this contradicts with the way we pick $I$ in option~\ref{sim:mem}.

                \item The search finds a not-done store $S$: 
                Using the same argument as above, there cannot be any store for $a$ between $I$ and $S$ in ROB.
                We now prove that the data of the $S$ must have been computed (for invariant~\ref{inv:done}).
                We prove by contradiction, i.e., we assume the data of $S$ is not yet computed.
                We backtrack the dependency chain on the data of $S$, and get $I_k \DataDep I_{k-1} \DataDep \ldots \DataDep I_1 \DataDep S$, where $I_k$ is a not-done load and $I_1\ldots I_{k-1}$ are all not-done reg-to-reg instructions.
                Since there is no store for $a$ between $S$ and $I$, we have $I_k \PreservePO I_1 \DepOrd I \Rightarrow I_k \MemOrd i$, contradicting with the way we pick $i$ in option~\ref{sim:mem} ($j_n$ is a not-done load).
                
                Since the data of $S$ has been computed, $I$ reads from $S$.
                We now need to verify invariant~\ref{inv:result}.
                Since $S$ is not-done, we have $I \MemOrd S$, i.e., the Load-Value axiom can only select from stores $\ProgOrd I$.
                Since there is no other store for $a$ between $S$ and $I$ in the ROB, the Load-Value axiom also agrees on $S\ReadFrom I$.
                
                \item The search finds nothing:
                In this case, $I$ reads from $m[a]$, and we need to verify invariant~\ref{inv:result}.
                We first show that all stores for $a$ older than $I$ in ROB are done.
                If there are not-done stores for $a$ older than $I$ in the ROB, then let $S$ be the youngest one among them.
                There cannot be any done store for $a$ between $S$ and $I$, because the guard of the Execute-Store rule that marks the store as done cannot be satisfied.
                The address of $S$ cannot be computed (otherwise the search will stop at $S$).
                Now we backtrack the dependency chain on the address of $S$ as we do in the first case, and can show a contradiction. 
                
                Assume $m[a]$ is last written by store $S^*$ before this rule fires.
                Thus, for any done store $S'$ for $a$ when this rule fires, either $S' = S^*$ or $S' \MemOrd S^*$.
                Since loads and stores can only be marked as done via option~\ref{sim:mem} in the operational model and $S^*$ is already done, we have $S^* \MemOrd I$.
                For any store $S_1 <_{mo} I$, $S_1$ must be done, so either $S_1 = S^*$ or $S_1 \MemOrd S^*$.
                For any store $S_2 <_{po} I$, $S_2$ is also done, so either $S_2 = S^*$ or $S_2 \MemOrd S^*$.
                Therefore, the Load-Value axiom also agrees on $S^*\ReadFrom I$.
            \end{enumerate}
        \end{itemize}
    \end{enumerate}
\end{proof}

\newpage
\section{Equivalence of COM and GAM}
\label{sec:gamcom}

We first define one more derived relation:
\begin{itemize}
  \item Reads-from internal ($\Rfi$), which is the subset of $\ReadFrom$ for which both the read and the write are in the same thread
\end{itemize}

\subsection{GAM $\subseteq$ COM}

\begin{lemma}\label{lem:app_com_in_memord}
  All of $\Rfe$, $\Coh$, $\Fr$, and $\PreservePO$ are contained in $\MemOrd$.
\end{lemma}
\begin{proof}
  Two of the four cases are easy: $\Coh$ is contained in $\MemOrd$ by construction, and $\PreservePO$ is contained in $\MemOrd$ by the Inst-Order axiom.

  By the Load-Value axiom, if for any write $w$ and read $r$, if $w\ReadFrom r$, then $w$ precedes $r$ either in $\ProgOrd$ or in $\MemOrd$.  The former is ruled out in the definition of $\Rfe$, and hence $w$ must precede $r$ in $\MemOrd$.

  The proof for $\Fr$ proceeds by contradiction.  Suppose there is some write $w$ and some read $r$ such that $r\Fr w$ and $w\MemOrd r$.  Then by definition of $\Fr$, there is some other write $w'$ such that $w'\ReadFrom r$ and $w'\Coh w$.  Furthermore, since $\Coh\subseteq\MemOrd$, we have $w'\MemOrd w\MemOrd r$.  This, however, contradicts the Load-value axiom, as $w'$ is not the $\MemOrd$-maximal candidate write.
\end{proof}

The SC-per-Location axiom will take a bit more work to prove.
To start, define $\Eco$ as the union of the following relations:
\begin{itemize}
  \item $\Coh$ (Write to Write)
  \item $\Fr$ (Read to Write)
  \item ${\Coh}^*;\ReadFrom$ (Write to Read)
  \item $\RfInv;{\Coh}^*;\ReadFrom$ (Read to Read)
\end{itemize}

\begin{lemma}\label{lem:app_eco_either}
  For all pairs $i_1$, $i_2$ of memory accesses to the same address, either $i_1\Eco i_2$ or $i_2\Eco i_1$.
\end{lemma}
\begin{proof}
  By construction.  All pairs of same-address writes are ordered in $\Coh$ by definition.  For any read $r$ and write $w$, let $w'$ be the write such that $w'\ReadFrom r$.  Then either:
  \begin{itemize}
    \item $w=w'$, so $w\ReadFrom r$, and hence $w{\Coh}^*;\ReadFrom r$,
    \item $w\Coh w'$, so $w\Coh;\ReadFrom r$, and hence $w{\Coh}^*;\ReadFrom r$, or
    \item $w'\Coh w$, so $r\RfInv;\Coh w$, and $r\Fr w$.
  \end{itemize}
  Likewise, for any two reads $r_1$ and $r_2$, let $w_1$ and $w_2$ be the writes such that $w_1\ReadFrom r_1$ and $w_2\ReadFrom r_2$.  Then either:
  \begin{itemize}
    \item $w_1=w_2$, so $r_1\RfInv;\ReadFrom r_2$, and hence $r_1\RfInv;{\Coh}^*;\ReadFrom r_2$,
    \item $w_1\Coh w_2$, so $r_1\RfInv;\Coh;\ReadFrom r_2$, and hence $r_1\RfInv;{\Coh}^*;\ReadFrom r_2$,
    \item $w_2\Coh w_1$, so $r_2\RfInv;\Coh;\ReadFrom r_1$, and hence $r_2\RfInv;{\Coh}^*;\ReadFrom r_1$.
  \end{itemize}
\end{proof}

If $i_1$ and $i_2$ are related in program order, then the $\Eco$ direction must match:
\begin{lemma}\label{lem:app_eco_poloc}
  If $i_1\PoLoc i_2$, then $i_1\Eco i_2$.
\end{lemma}
\begin{proof}
  By Lemma \ref{lem:app_eco_either}, either $i_1\Eco i_2$ or $i_2\Eco i_1$.  We show that the latter always results in a contradiction (except for one case in which it overlaps with the former).
  \begin{itemize}
    \item If $i_1$ and $i_2$ are both writes, then $i_1\PoLoc i_2\Coh i_1$, so $i_1\PoLoc i_2\MemOrd i_1$, which contradicts Definition~\ref{def:ppo-same-addr}.\ref{ppo:st->st}.
    \item If $i_1$ is a read and $i_2$ is a write, then suppose $i_2{\Coh}^*~i\ReadFrom i_1$ for some $i$.  If $i_2=i$, then by the LoadValue axiom, either $i_2\PoLoc i_1$, which contradicts the hypothesis, or $i_2\MemOrd i_1$, which contradicts Definition~\ref{def:ppo-same-addr}.\ref{ppo:ld->st}.
    Therefore, suppose $i_2\Coh i\ReadFrom i_1$.  If $i\Rfe i_1$, then $i_1\PreservePO i_2\Coh i\Rfe i_1$ by Definition~\ref{def:ppo-same-addr}.\ref{ppo:ld->st}, which contradicts Causality.
    If $i\Rfi i_1$, then $i\PoLoc i_1$; otherwise, by the LoadValue axiom, $i_1\MemOrd i$, which contradicts Definition~\ref{def:ppo-same-addr}.\ref{ppo:ld->st}. 
    However, this means $i\PoLoc i_2\Coh i$, which again contradicts Definition~\ref{def:ppo-same-addr}.\ref{ppo:st->st}.
    \item If $i_1$ is a write and $i_2$ is a read, then suppose $i_2\Fr i_1$, and let $i$ be the write such that $i\ReadFrom i_2$ and $i\Coh i_1$.  Since $i\MemOrd i_1$, $i$ is not the $\MemOrd$-maximal store from which $i_2$ should read, and the LoadValue axiom is violated.
    \item If $i_1$ and $i_2$ are both reads, then suppose $i_2\RfInv i_3\Coh i_4\ReadFrom i_1$ for some $i_3$ and $i_4$.  (The case $i_2\RfInv;\ReadFrom i_1$ implies $i_1\Eco i_2$.)  Then $i_2\Fr i_4$.  If $i_4\Rfi i_1$, then $i_4\PoLoc i_2\Fr i_4$, which as we have already seen in the previous case is forbidden.  If $i_4\Rfe i_1$, then either there is some write $i_5$ such that $i_1\PoLoc i_5\PoLoc i_2$, or there is no such write.  If $i_5$ exists, then $i_4\Coh i_5$, $i_3\Coh i_5$, and $i_2\Fr i_5\PoLoc i_2$, which as we have already seen is forbidden.  If $i_5$ does not exist, then $i_1\PreservePO i_2\Fr i_4\Rfe i_1$, which contradicts Causality.
  \end{itemize}
\end{proof}

\begin{theorem}
  The SC-per-Location axiom is satisfied.
\end{theorem}
\begin{proof}
  First, by Lemma~\ref{lem:app_eco_poloc}, all $\PoLoc$ edges involving at least one write can be converted into sequences containing only $\ReadFrom$, $\Coh$, and $\Fr$.  So we consider only cycles with $\ReadFrom$, $\Coh$, $\Fr$, and read-to-read $\PoLoc$ edges.  Among such cycles, first consider cycles with no $\Coh$ or $\Fr$ edges.  Such cycles cannot contain $\ReadFrom$ either, because neither $\ReadFrom$ nor read-read $\PoLoc$ edges can end at a write node, and so there cannot be a source for $\ReadFrom$ relations.  This leaves a cycle consisting only of $\PoLoc$, which is a contradiction.

  Now, consider cycles with at least one $\Coh$ or $\Fr$ edge.  Replace every instance of read-read $\PoLoc$ in the cycle with $\RfInv;{\Coh}^*;\ReadFrom$ per Lemma~\ref{lem:app_eco_poloc}.  Now, because $\Coh$ and $\Fr$ both target writes, every appearance of $\RfInv$ must be preceded either by $\ReadFrom$ or by $\RfInv;{\Coh}^*;\ReadFrom$.  In particular, every appearance of $\RfInv$ must be preceded directly by $\ReadFrom$.  Since $\ReadFrom;\RfInv$ is the identity function, all appearances of $\RfInv$ in the cycle can be eliminated by simply removing each $\ReadFrom;\RfInv$ pair in the cycle.  This leaves a cycle with only $\ReadFrom$, $\Coh$, and $\Fr$.
  If there are any reads in such a cycle, then by similar logic as above, the incoming relation must be $\ReadFrom$ and the outgoing relation must be $\Fr$, but this pair is equivalent to a $\Coh$ edge between writes alone.  Repeating such a transformation produces a cycle consisting of only $\Coh$.  Since by hypothesis there is at least one such $\Coh$ edge, this is a contradiction.
\end{proof}

\newpage
\section{Alloy Model for Empirical Validation}
\label{sec:alloy}

Figure~\ref{fig:alloy} shows the Alloy model used for validation.

\lstdefinelanguage{alloy}{
  morekeywords={abstract, sig, extends, pred, fun, fact, no, set, one, lone, let, not, all, iden, some, run, for},
  morecomment=[l]{//},
  morecomment=[s]{/*}{*/},
  commentstyle=\color{green!40!black},
  keywordstyle=\color{blue!40!black},
  moredelim=**[is][\color{red}]{@}{@},
  escapeinside={!}{!},
}
\lstset{language=alloy}
\lstset{aboveskip=0pt}
\lstset{belowskip=0pt}

\begin{figure}[h!]
  \tt\bfseries\centering\footnotesize
  \begin{lstlisting}
// Model of memory
sig Address {}
abstract sig Event {
  po: lone Event, ppo: set Event, mo: set Event, address: one Address }
sig Read extends Event {}
sig Write extends Event { rf: set Read }
fun po_loc : Event->Event { ^po & address.~address }
fact { acyclic[po] }
fact { rf.~rf in iden }
fact { total[mo, Event] } // definition of total omitted for space
fact { (Write <: po_loc :> Write) + (Read <: po_loc :> Write)
  + (Read <: (po_loc - (po_loc.po_loc)) :> Read) in ppo }

// GAM axioms
fun candidates[r: Read] : set Write {
  (r.~mo & Write & r.address.~address) // writes preceding r in <mo
  + (r.^~po & Write & r.address.~address)} // writes preceding r in <po
pred InstOrder { ppo in mo }
pred LoadValue { all w: Write | all r: Read |
  w->r in rf <=> w in (let c = candidates[r] | c - c.~mo)} // i.e., max_<mo
pred GAM { InstOrder and LoadValue }

// COM axioms
fun rfe : Write->Read { rf - (^po + ^~po) }
fun co : Write->Write { Write <: ((address.~address) & mo) :> Write }
fun fr : Read->Write { ~rf.co + ((Read - Write.rf) <: address.~address :> Write) }
pred SC_per_Location  { acyclic[rf + co + fr + po_loc] }
pred Causality { acyclic[rfe + co + fr + ppo] }  // def. of acyclic omitted for space
pred COM { SC_per_Location and Causality }

// Equivalence Checks
check gam_com { GAM => COM } for 7
check com_gam { rfe + co + fr + ppo in mo => COM => GAM } for 7
  \end{lstlisting}
  \caption{Comparing GAM and COM in Alloy}
  \label{fig:alloy}
\end{figure}

\newpage
\section{Completeness: GAM-I2E Axiomatic Model $\subseteq$ GAM-I2E Operational Model}\label{sec:i2_op_contain_axi}
\begin{theorem}
    GAM-I2E axiomatic model $\subseteq$ GAM-I2E operational model.
\end{theorem}
\begin{proof}
The goal is that for any legal axiomatic relations $\langle \ProgOrd, \MemOrd, \ReadFrom\rangle$ (which satisfy the GAM-I2E axioms), we can run the GAM-I2E operational model to simulate the same program behavior.
In each step of the simulation, we first decide which rule to fire in the operational model based on the current state of the operational model and $\MemOrd$, and then we fire that rule.
Here is the algorithm to determine which rule to fire in each simulation step:
\begin{enumerate}
    \item If in the operational model there is a processor whose next instruction is not a load, we fire an Execute-Reg-Branch or Execute-Store-Fence rule to execute that instruction in the operational model.
    \item If the above case does not apply, and in the operational model there is a fence that can be dequeued from the local buffer, then we fire the Dequeue-Fence rule to dequeue that fence in the operational model.
    \item \label{sim:i2e:st} If neither of the above cases applies, and in the operational model there is a store $S$ in the local buffer of a processor, and $S$ can be dequeued from the local buffer (i.e., the guard for the Dequeue-Store rule is true), and all stores before $S$ in $\MemOrd$ are already in $\IIEMemOrd$, then we fire a Dequeue-Store rule to dequeue $S$ in the operational model.
    \item \label{sim:i2e:ld} If none of the above cases applies, then in the operational model there must be a processor such that the next instruction of the processor is a load $L$, and $L$ can be executed (i.e., the guard for the Execute-Load rule is true), and all stores before $L$ in $\MemOrd$ are already in $\IIEMemOrd$.
    We fire an Execute-Load rule to execute $L$ in the operational model.
    In the Execute-Load rule of $L$, we insert $L$ into $\IIEMemOrd$ such that for any instruction $I$ already in $\IIEMemOrd$, if $I \MemOrd L$ then $I \IIEMemOrd L$, otherwise $L \IIEMemOrd I$.
\end{enumerate}
After each step of the simulation, we keep the following invariants:
\begin{enumerate}
    \item \label{inv:i2e:po} The execution order on each processor is a prefix of the $\ProgOrd$ of that processor.
    \item The result of each executed instruction is the same as that in $\ProgOrd$.
    \item \label{inv:i2e:rf} The store read by each executed load is the same as that indicated by the $\ReadFrom$ edges.
    \item \label{inv:i2e:stuck} The simulation cannot get stuck.
    \item \label{inv:i2e:mo-match} For two memory instruction $I_1$ and $I_2$, if $I_1 \IIEMemOrd I_2$ in the operational model, then $I_1 \MemOrd I_2$ in the axiomatic relations.
    \item \label{inv:i2e:mo-prefix} The order of all stores in $\IIEMemOrd$ is a prefix of the order of all stores in $\MemOrd$.
\end{enumerate}
The first two induction invariants imply that \emph{before} each simulation step, the following properties hold for each processor $i$ (assuming the next instruction of the processor is $I$):
\begin{enumerate}
    \item $\IIEProgOrd$ is a prefix of $\ProgOrd$ (of processor $i$) up to $I$ (including $I$).
    \item For any instructions $I_1 \PreservePO I_2$ from processor $i$, if $I_1$ and $I_2$ are not ordered after $I$ in $\ProgOrd$ (i.e., $I_2$ may be equal to $I$), then $I_1 \IIEPreservePO I_2$.
    \item \label{prop:i2e:ppo} For any instructions $I_1$ and $I_2$, if $I_1 \IIEPreservePO I_2$, then $I_1 \PreservePO I_2$.
\end{enumerate}

Now we examine each case in the simulation algorithm and prove that all invariants hold:
\begin{enumerate}
    \item We execute a non-load instruction: trivial.
    \item We dequeue a fence from the local buffer: trivial.
    \item We dequeue a store $S$ from the local buffer: In this case, we need to verify invariants \ref{inv:i2e:mo-match} and \ref{inv:i2e:mo-prefix}.
    Invariant~\ref{inv:i2e:mo-prefix} is trivial, because all stores older than $S$ in $\MemOrd$ are already in $\IIEMemOrd$ (as required by case~\ref{sim:i2e:st} in the algorithm).
    We now consider invariant~\ref{inv:i2e:mo-match}.
    For each instruction $I$ already in $\IIEMemOrd$ at the dequeue time, $I$ must be added to $\IIEMemOrd$ by case~\ref{sim:i2e:st} or \ref{sim:i2e:ld} in the simulation algorithm.
    Since these two cases require that every store older than $I$ in $\MemOrd$ to be present in $\IIEMemOrd$, $S$ cannot be older than $I$ in $\MemOrd$, i.e., $I \MemOrd S$.
    \item We execute a load $L$:
    We first need to verify invariant~\ref{inv:i2e:stuck}, i.e., we are able to find such an $L$ that satisfies the requirements in case~\ref{sim:i2e:ld} of the simulation algorithm.
    We prove this by contradiction, i.e., such $L$ cannot be found.
    In this case, the next instruction of every processor is a load.
    We examine why the next instruction $L_1$ (which is a load) of processor 1 does not satisfy the requirements of case~\ref{sim:i2e:ld} of the simulation algorithm.
    There are two possibilities:
    \begin{enumerate}
        \item There is a store $S_2 \MemOrd L_1$ but $S_2$ is not yet in $\IIEMemOrd$.
        \item The guard of the Execute-Load rule for $L_1$ is false.
        We backtrack which instruction is stalling $L_1$.
        There must exist an instruction $I_1$ in the local buffer of processor 1 which is ordered before $L_1$ in $\IIEPreservePO$.
        If $I_1$ is a fence, then $I_1$ cannot be dequeued because there is another instruction $I_2 \IIEMemOrd I_1$ in the local buffer.
        We keep doing this until we find a store, i.e., $I_k \IIEPreservePO I_{k-1} \IIEPreservePO \cdots \IIEPreservePO I_1 \IIEPreservePO L_1$, where $I_1\cdot I_k$ are all in the local buffer of processor 1, $I_1\cdots I_{k-1}$ are fences, and $I_k$ is a store.
        According to property~\ref{prop:i2e:ppo}, $I_k\IIEPreservePO L_1\Rightarrow I_k \PreservePO L_1 \Rightarrow I_k \MemOrd L_1$.
    \end{enumerate}
    In either case, we find a store $S_2 \MemOrd L_1$, and $S_2$ is not in $\IIEMemOrd$.
    Now we consider why $S_2$ is not in $\IIEMemOrd$, and there are two possibilities:
    \begin{enumerate}
        \item $S_2$ is not executed yet:
        Assume $S_2$ is in processor $i$ in $\ProgOrd$.
        The next instruction to execute in the processor of $S_2$ in the operational model must be a load $L_3$.
        According to invariant~\ref{inv:i2e:po}, since $S_2$ is not in the prefix of $\ProgOrd$ of processor $i$ up to $L_3$, we have $L_3 \ProgOrd S_2 \Rightarrow L_3 \LSFOrd S_2 \Rightarrow L_3 \MemOrd S_2$.
        Following the previous argument, $L_3$ cannot be executed because of a store $S_3$ which is before $L_3$ in $\MemOrd$ but is not in $\IIEMemOrd$.
        That is, $S_3 \MemOrd S_2$ and $S_3$ is not in $\IIEMemOrd$.
        \item $S_2$ is the local buffer of processor $i$:
        There are two possible reasons that stops $S_2$ from being dequeued:
        \begin{enumerate}
            \item There is a store $S_3 \MemOrd S_2$ and $S_3$ is not in $\IIEMemOrd$.
            \item The guard of the Dequeue-Store rule is false.
            Using the previous argument, there must be a store  $S_3 \MemOrd S_2$, and $S_3$ is in the local buffer.
        \end{enumerate}
    \end{enumerate}
    In all cases, we can find a store $S_3 \MemOrd S_2$, and $S_3$ is not in $\IIEMemOrd$.
    Since the simulation algorithm is assumed to get stuck, we can keep doing this, and find $S_k \MemOrd S_{k-1} \MemOrd \cdots \MemOrd S_1 \MemOrd L_1$, where $S_1\ldots S_k$ are all stores that are not in $\IIEMemOrd$, and $k$ can be infinitely large.
    However, there can only be finite number of stores before $L_1$ in $\MemOrd$.
    Therefore, we must be able to find a load $L$ that satisfies the requirements of case~\ref{sim:i2e:ld} of the simulation algorithm.
    
    We also need to verify that $L$ can indeed be inserted into $\IIEMemOrd$ as instructed in case~\ref{sim:i2e:ld} of the simulation algorithm.
    Since both $\MemOrd$ and $\IIEMemOrd$ are total orders, invariant~\ref{inv:i2e:mo-match} ensures that we can cut $\IIEMemOrd$ into two parts, i.e., one part is before $L$ in $\MemOrd$ and the other part is after $L$ in $\MemOrd$.
    Then we can simply place $L$ at the cutting point of $\IIEMemOrd$.
    This also ensures that invariant~\ref{inv:i2e:mo-match} will still hold after this step of simulation.
    
    Finally we need to show that invariant~\ref{inv:i2e:rf} still holds.
    Assume $L$ is from processor $i$, loads address $a$, and reads from store $S$ in the Execute-Load rule.
    Consider where $S$ resides when we fire the Execute-Load rule:
    \begin{enumerate}
        \item $S$ is in the local buffer of processor $i$:
        Since all stores $\MemOrd L$ are already in $\IIEMemOrd$, $S$ does not precede $L$ in $\MemOrd$, i.e., $L\MemOrd S$.
        Invariant~\ref{inv:i2e:po} implies that $S\ProgOrd L$.
        Therefore the Load-Value axiom can only select stores $\ProgOrd L$ as the source for the load result.
        Since all stores for the same address in the same processor are ordered by $\SameAddrOrd$ and thus $\MemOrd$, the Load-Value axiom will pick the youngest store in $\ProgOrd$ among all stores that is before $L$ in $\ProgOrd$.
        Since $S$ is the most recently executed store for $a$ in processor $i$, invariant~\ref{inv:i2e:po} ensures that $S$ is the store picked by the Load-Value axiom.
        
        \item $S$ is already in $\IIEMemOrd$:
        In this case, the local buffer of processor $i$ cannot have any store for address $a$.
        Invariant~\ref{inv:i2e:po} says that for any store $S'$ for $a$ which is ordered before $L$ in $\ProgOrd$, $S'$ must have been executed in processor $i$ in the operational model.
        Therefore, $S'$ must be already in $\IIEMemOrd$.
        The way we find  $L$ ensures that for any store $S''$ for $a$ that is ordered before $L$ in $\MemOrd$, $S''$ must be in $\IIEMemOrd$.
        Thus, all stores that are visible to $L$ according to the Load-Value axiom are all in $\IIEMemOrd$ now.
        Invariants~\ref{inv:i2e:mo-match} and ref{i2e:mo-prefix} both say that the orderings between all such $S'$ and $S''$ are the same in $\IIEMemOrd$ and $\MemOrd$.
        Since the Execute-Load rule uses the same way as the Load-Value axiom to determine the load value, invariant~\ref{inv:i2e:rf} must hold.
    \end{enumerate}
\end{enumerate}
\end{proof}

\end{document}